\documentclass{amsart}

\usepackage{amssymb,latexsym,amsfonts,amsmath}
\usepackage{graphicx}
\usepackage{paralist}
\usepackage{dsfont}
\usepackage{booktabs}
\usepackage{subfigure}
\usepackage{eso-pic}
\usepackage{epsfig}
\usepackage{mathtools}
\usepackage {url}
\usepackage{epstopdf}
\usepackage{tikz}
\DeclareSymbolFont{symbolsC}{U}{pxsyc}{m}{n}
\SetSymbolFont{symbolsC}{bold}{U}{pxsyc}{bx}{n}
\DeclareFontSubstitution{U}{pxsyc}{m}{n}
\DeclareMathSymbol{\medcirc}{\mathbin}{symbolsC}{7}
\usetikzlibrary{arrows,calc,decorations.markings,math,arrows.meta}

\usepackage{algorithm}
\usepackage[noend]{algpseudocode}
\algnewcommand\algorithmicinput{\textbf{Input:}}
\algnewcommand\Input{\item[\algorithmicinput]}
\algnewcommand\algorithmicoutput{\textbf{Output:}}
\algnewcommand\Output{\item[\algorithmicoutput]}
\algnewcommand{\Initialize}[1]{%
  \State \textbf{Initialize:}
  \Statex \hspace*{\algorithmicindent}\parbox[t]{.8\linewidth}{\raggedright #1}
}

\newtheorem{theorem}{Theorem}[section]
\newtheorem{lemma}[theorem]{Lemma}

\newtheorem{problem}[theorem]{Problem}

\newtheorem{corollary}[theorem]{Corollary}

\newtheorem{definition}[theorem]{Definition}
\newtheorem{example}[theorem]{Example}

\newtheorem{remark}[theorem]{Remark}
\newtheorem{assumption}[theorem]{Assumption}
\numberwithin{equation}{section}

\newcommand{\R}{{\mathbb{R}}}

\newcommand{\N}{{\mathbb{N}}}

\newcommand{\ie}{{\it i.e.}}
\newcommand{\eg}{{\it e.g.}}

\newcommand{\argmin}{\textrm{arg}\min}

\newcommand{\Sys}{{dt-SCS\,\,}}

\usepackage{dsfont}

\begin{document}

\begin{abstract}
This paper focuses on synthesizing control policies for discrete-time stochastic control systems together with a lower bound on the probability that the systems satisfy the complex temporal properties. The desired properties of the system are expressed as linear temporal logic (LTL) specifications over finite traces. In particular, our approach decomposes the given specification into simpler reachability tasks based on its automata representation. We then propose the use of so-called \emph{control barrier certificate} to solve those simpler reachability tasks along with computing the corresponding controllers and probability bounds. Finally, we combine those controllers to obtain a hybrid control policy solving the considered problem. Under some assumptions, we also provide two systematic approaches for uncountable and finite input sets to search for control barrier certificates. We demonstrate the effectiveness of the proposed approach on a room temperature control and lane-keeping of a vehicle modeled as a four-dimensional single-track kinematic model. We compare our results with the discretization-based methods in the literature.

\end{abstract}

\title[Formal Synthesis of Stochastic Systems via Control Barrier Certificates]{Formal Synthesis of Stochastic Systems via Control Barrier Certificates}

\author[P. Jagtap]{Pushpak Jagtap$^1$} 
\author[S. Soudjani]{Sadegh Soudjani$^2$} 
\author[M. Zamani]{Majid Zamani$^{3,4}$} 

\address{$^1$Department of Electrical and Computer Engineering, Technical University of Munich, Germany.}
\email{pushpak.jagtap@tum.de}
\address{$^2$School of Computing, Newcastle University, United Kingdom.}
\email{Sadegh.Soudjani@newcastle.ac.uk}
\address{$^3$Computer Science Department, University of Colorado Boulder, USA.}
\email{majid.zamani@colorado.edu}
\address{$^4$Computer Science Department, Ludwig Maximilian University of Munich, Germany.}
\maketitle
\section{Introduction}
Formal synthesis of controllers for complex dynamical systems against complex specifications has gained significant attentions in the last decade \cite{tabuada2009verification,belta2017formal}. 
These specifications are usually expressed using temporal logic formulae or automata on (in)finite strings.
The synthesis problem is very challenging for systems that have continuous state spaces and are affected by uncertainties. The problem does not admit closed-form solutions in general and is hard to be solved exactly on such systems.

There have been several results in the literature utilizing approximate finite models as abstractions of the original \emph{stochastic} dynamical systems for the formal policy synthesis. Existing results include policy synthesis for discrete-time stochastic hybrid systems \cite{abate2008probabilistic,MMS2020,HS_TAC20}, control of switched discrete-time stochastic systems \cite{7029024}, and symbolic control of incrementally stable stochastic systems \cite{zamani2014symbolic}.
These approaches rely on the discretization of the state set together with a formal upper-bound on the approximation error. These approaches suffer severely from the curse of dimensionality (\ie, computational complexity grows exponentially with the dimension of the state set). To alleviate this issue, sequential gridding \cite{esmaeil2013adaptive}, discretization-free abstractions \cite{zamani2017towards,jagtap2017symbolic}, and compositional abstraction-based techniques \cite{esmaeilzadehsoudjani,lavaei2018compositional} are proposed under suitable assumptions on the system dynamics (\eg, Lipschitz continuity or incremental input-to-state stability).

{For \emph{non-stochastic} systems, discretization-free approaches based on barrier certificates were proposed for verification and synthesis to ensure safety \cite{7782377, 8430747,nilsson2018barrier, prajna2006barrier, wieland2007constructive}. 
	The authors in \cite{7364197} generalize the idea of the barrier certificate by combining it with the automata representation of LTL specifications for the verification of temporal property for nonlinear non-stochastic systems. The work is then extended for the verification of hybrid dynamical systems against syntactically co-safe LTL specifications \cite{bisoffi2018hybrid} and for the synthesis of an online control strategy for multi-agent systems enforcing LTL specifications \cite{srinivasan2018control}. 
	There are a few recent results using barrier certificates on non-stochastic systems to satisfy more general specifications. Results include the use of time-varying control barrier functions to satisfy signal temporal logic \cite{8404080} and control barrier certificate to design policies for reach and stay specification for non-stochastic switched systems \cite{Ravanbakhsh2017ACO}. Most of the synthesis results mentioned above consider prior knowledge of barrier certificates to provide online control strategies using quadratic programming. These results may not be suitable while dealing with constrained input sets which is the case in almost all real world applications.} 

{For \emph{stochastic} systems, there are very few works available in the literature to synthesize controllers against complex specifications using discretization-free approaches. The results include the synthesis of controller for continuous-time stochastic systems enforcing syntactically co-safe LTL specifications \cite{horowitz2014compositional}, where the authors use automata representation corresponding to the specifications to guide a sequence of stochastic optimal control problems.
	The paper \cite{farahani2018shrinking} considers synthesis for ensuring a lower bound on the probability of satisfying a specification in signal temporal logic. It encodes the requirements as chance constraints and inductively decomposes them into deterministic inequalities using the structure of the specification. Barrier certificates are utilized in \cite{huang2017probabilistic,steinhardt2012finite,4287147} for verification of stochastic (hybrid) systems but only with respect to the invariance property.}

Our recent results in \cite{jagtap2018temporal} present the idea of combining automata representation of a complex specification and barrier certificates, 
for formal \emph{verification} of stochastic systems without requiring any stability assumption on the dynamics of the system.
The current manuscript follows a similar direction to solve the problem of formal synthesis for stochastic systems.


To the best of our knowledge, this paper is the first to utilize the notion of control barrier certificates for the synthesis of discrete-time stochastic control systems against complex temporal logic specifications.
We consider temporal properties expressed in a fragment of LTL formulae, namely, LTL on finite traces, referred to as LTL$_F$ \cite{6942758}.
We provide a systematic approach to synthesize an {offline} control policy together with a lower bound on the probability that the LTL$_F$ property is satisfied over finite-time horizon.
{To achieve this, we utilize the notion of control barrier certificates which in general can only provide an \emph{upper bound} on the reachability probability. Since we are looking for a \emph{lower bound}, we first take the negation of the LTL$_F$ specification and decompose satisfaction of the negation into a sequence of simpler reachability tasks based on the structure of the automaton associated with the negation of the specification. Then, controllers and corresponding upper bounds are obtained for these simplified reachability tasks with the help of control barrier certificates. In the final step, we combine these controllers and probability bounds to provide a hybrid control policy and a lower bound on the probability of satisfying the original LTL$_F$ property.}

In general, there is no guarantee that barrier certificates exist for a given stochastic system. Even if we know one exists, there is no complete algorithm for its computation. In this paper, we provide two systematic approaches to search for control barrier certificates under suitable assumptions on the dynamics of the system and the shape of the potential barrier certificates. The first approach utilizes sum-of-square optimization technique \cite{parrilo2003semidefinite} and is suitable for dynamics with continuous input sets and polynomial dynamics. The second approach uses the counter-example guided inductive synthesis (CEGIS) scheme which is adapted from \cite{ravanbakhsh2015counter,Ravanbakhsh2017ACO} and is suitable for systems with finite input sets.

The remainder of this paper is structured as follows. In Section~\ref{sec_prelim}, we introduce discrete-time stochastic control systems and the linear temporal logic over finite traces. Then, we formally defined the problem considered in this paper. We discuss in Section~\ref{Sec_CBC} the notion of control barrier certificates and results for the computation of upper bound on the probability of satisfying reachability specifications. Section~\ref{runs} provides an algorithm to decompose LTL$_F$ specification into sequential reachability using deterministic finite automaton (DFA) corresponding to specification. In Section~\ref{BC_computation}, we provide results on the synthesis of control policy together with the lower bound on the probability of satisfaction of LTL$_F$ specifications using control barrier certificates. It also provides systematic approaches to search for control barrier certificates. 
Section~\ref{sec_case} demonstrates the effectiveness of the results on two case studies: (i) room temperature control and (ii) lane keeping of a vehicle. Finally, Section~\ref{sec_conclusion} concludes the paper.




\section{Preliminaries}\label{sec_prelim}
\subsection{Notations}
We denote the set of nonnegative integers by $\N_0 := \{0, 1, 2, \ldots\}$ and the set of positive integers by $\N := \{1, 2, 3, \ldots \}$. The symbols $ \R$, $\R^+,$ and $\R_0^+ $ denote the set of real, positive, and nonnegative real numbers, respectively. We use $ \mathbb{R}^{n\times m} $ to denote the space of real matrices with $ n $ rows and $ m $ columns. For a finite set $A$, we denote its cardinality by $|A|$. The logical operators `not', `and', and `or' are denoted by $\neg$, $\wedge$, and $\vee$, respectively.

We consider a probability space with the tuple $ (\Omega,\mathcal{F}_\Omega,\mathbb{P}_\Omega) $, where $\Omega$ is the sample space, $\mathcal{F}_\Omega$ is a sigma-algebra on $\Omega$ comprising the subset of $\Omega$ as events, and $\mathbb{P}_\Omega$ is a probability measure that assigns probabilities to events. {We assume that random variables introduced in this article are measurable functions of the form $X:(\Omega,\mathcal{F}_\Omega)\rightarrow(S_X,\mathcal{F}_X)$ mapping measurable space $(\Omega,\mathcal{F}_\Omega)$ to another measurable space $(S_X,\mathcal{F}_X)$ and assigns probability measure to $(S_X,\mathcal{F}_X)$ according to $Prob\{A\}=\mathbb{P}_\Omega\{X^{-1}(A)\}$ for any $A\in\mathcal{F}_X$.
In words, $S_X$ is the domain of the random variable $X$ and $\mathcal{F}_X$ is a collection of subsets of $S_X$ such that $X$ assigns probability to the elements of this collection.} We often directly discuss the probability measure on $(S_X,\mathcal{F}_X)$ without explicitly mentioning the underlying probability space and the function $X$ itself.  


\subsection{Discrete-time stochastic control systems}
\label{Sec_dtSCS}
{In this work, we consider discrete-time stochastic control systems (dt-SCS) that are extensively employed as models of systems under uncertainty in economics and finance \cite{Ev_Ar87} and in many engineering systems \cite{BS96}. Examples of using {dt-SCS} include modeling inventory-production systems \cite{hll1996}, demand response in energy networks \cite{soudjani2014formal}, and analyzing max-plus linear systems in transportation \cite{7335578}.}

A \Sys is
given by the tuple $\mathfrak S=(X,V_{\textsf w},U,w,f)$, where {$X$ is the state set, $V_{\textsf w}$ is the uncertainty set, and $U$ is the input set of the system}. We denote by $(X,\mathcal{B}(X))$ the measurable space with $\mathcal{B}(X)$ being the Borel sigma-algebra on the state space. Notation $w$ denotes a sequence of independent and identically distributed (i.i.d.) random variables on the set $V_{\textsf w}$ as $w:=\{w(k):\Omega\rightarrow V_{\textsf w}, \ k\in\mathbb N_0\}$.  
The map $f : X \times U \times V_{\textsf w} \rightarrow X$ is a measurable function characterizing the state evolution of the system. For a given initial state $x(0)\in X$, the state evolution can be written as
\begin{equation}
 x(k+1)=f(x(k),u(k),w(k)), \ \ \ k\in\mathbb{N}_0.
\label{DTSCS}
\end{equation}
We are interested in synthesizing a control policy $\rho$ that guarantees a potentially tight lower bound on the probability that the system $\mathfrak S$ satisfies a specification expressed as a temporal logic property. The syntax and semantics of the class of specifications dealt with in this paper are provided in the next subsection. In this work, we consider \textit{history-dependent policies} given by $\rho=(\rho_0,\rho_1,\ldots,\rho_n,\ldots)$ with functions $\rho_n: H_n\rightarrow U$, where $H_n$ is a set of all $n$-histories $h_n$ defined as $h_n:=(x(0),u(0),x(1),u(1),\ldots,x(n-1),u(n-1),x(n))$. A subclass of policies are called \emph{stationary} and are defined as $\rho=(u,u,\ldots,u,\ldots)$ with a function $u: X\rightarrow U$. In stationary policies, the mapping at time $n$ depends only on the current state $x_n$ and does not change over time.

\subsection{Linear temporal logic over finite traces}
\label{Sec_LTL}
In this subsection, we introduce linear temporal logic over finite traces, referred to as LTL$_F$ \cite{de2013linear}, which will be used later to express temporal specifications for our synthesis problem. Properties LTL$_F$ use the same syntax of LTL over infinite traces given in \cite{baier2008principles}. The LTL$_F$ formulas over a set $ \Pi $ of atomic propositions are obtained as follows:
\begin{align*}
 \varphi ::=  &\textsf{ true} \mid p \mid \neg \varphi \mid \varphi_1 \wedge \varphi_2 \mid \varphi_1 \vee\varphi_2\mid\medcirc \varphi  \mid  \lozenge\varphi \mid \square\varphi \mid \varphi_1\mathcal{U}\varphi_2 ,
 \end{align*}
 where $p \in \Pi$, $\medcirc $ is the next operator, $\lozenge$ is eventually, $\square$ is always, and $\mathcal{U}$ is until. The semantics of LTL$_F$ is given in terms of \textit{finite traces}, \ie, finite words $\sigma$, denoting a finite non-empty sequence of consecutive steps over $\Pi$. We use $|\sigma |$ to represent the length of $\sigma$ and $\sigma_i$ as a propositional interpretation at the $i$th position in the trace, where $0\leq i < |\sigma |$. Given a finite trace $\sigma$ and an LTL$_F$ formula $\varphi$, we inductively define when an LTL$_F$ formula $\varphi$ is true at the $i$th step $(0\leq i <|\sigma |)$ and denoted by $\sigma,i\models\varphi$, as follows:
 \begin{itemize}
 \item $\sigma,i\models \textsf{true}$;
  \item $\sigma,i\models p$, for $p\in\Pi$ iff $p\in\sigma_i$;
   \item $\sigma,i\models \neg\varphi$ iff $\sigma,i\not\models\varphi$;
    \item $\sigma,i\models \varphi_1\wedge\varphi_2$ iff $\sigma,i\models\varphi_1$ and $\sigma,i\models\varphi_2$;
      \item $\sigma,i\models \varphi_1\vee\varphi_2$ iff $\sigma,i\models\varphi_1$ or $\sigma,i\models\varphi_2$;  
     \item $\sigma,i\models \medcirc \varphi$ iff $i<|\sigma |-1$ and $\sigma,i+1\models\varphi$; 
     \item $\sigma,i\models  \lozenge\varphi$ iff for some $j$ such that $i\leq j< |\sigma |$, we have $\sigma,j\models\varphi$; 
     \item $\sigma,i\models  \square\varphi$ iff for all $j$ such that $i\leq j<|\sigma |$, we have $\sigma,j\models\varphi$; 
     \item $\sigma,i\models  \varphi_1\mathcal{U}\varphi_2$ iff for some $j$ such that $i\leq j< |\sigma |$, we have $\sigma,j\models\varphi_2$, and for all $k$ s.t. $i\leq k<j$, we have $\sigma,k\models\varphi_1$. 
 \end{itemize}
 The formula $\varphi$ is true on $\sigma$, denoted by $\sigma\models\varphi$, if and only if $\sigma,0\models\varphi$. The set of all traces that satisfy the formula $\varphi$ is called the \emph{language} of formula $\varphi$ and is denoted by $\mathcal{L}(\varphi)$.  Notice that we also have the usual boolean equivalences such as $\varphi_1\vee\varphi_2\equiv \neg(\neg\varphi_1\wedge\neg\varphi_2)$, 
 $\varphi_1\implies\varphi_2\equiv \neg\varphi_1\vee\varphi_2$, $\lozenge\varphi \equiv \textsf{true } \mathcal{U} \varphi$, and $\square\varphi\equiv\neg\lozenge\neg\varphi$.
 
Next, we define deterministic finite automata which later serve as equivalent representations of LTL$_F$ formulae.
\begin{definition}\label{DFA1}
A deterministic finite automaton $($DFA$)$ is a tuple $\mathcal{A}=(Q,Q_0,\Sigma, \delta,F)$, where $Q$ is a finite set of states, $Q_0\subseteq Q$ is a set of initial states, $\Sigma$ is a finite set $($a.k.a. alphabet$)$, $\delta: Q\times\Sigma\rightarrow Q$ is a transition function, and $F\subseteq Q$ is a set of accepting states.
\end{definition}
We use notation $q\overset{\sigma}{\longrightarrow} q'$ to denote transition $(q,\sigma,q')\in\delta$.
A finite word $\sigma=(\sigma_0,\sigma_1,\ldots,\sigma_{n-1})\in \Sigma^n$ is accepted by DFA $\mathcal{A}$ if there exists a finite state run $q=(q_0,q_1,\ldots,q_{n})\in Q^{n+1}$ such that $q_0\in Q_0$, $q_i \overset{\sigma_i}{\longrightarrow} q_{i+1}$ for all $0\leq i< n$ and $q_{n}\in F$. The set of words accepted by $\mathcal{A}$ is called the accepting language of $\mathcal{A}$ and is denoted by $\mathcal{L}(\mathcal{A})$. We denote the set of successor states of a state $q\in Q$ by $\Delta(q)$.

The next result shows that every LTL$_F$ formula can be accepted by a DFA.

\begin{theorem}[\cite{zhu2019first,de2015synthesis}]\label{DFA}
Every LTL$_F$ formula $\varphi$ can be translated to a DFA $\mathcal{A}_\varphi$
that accepts the same language as $\varphi$, \ie, $\mathcal{L}(\varphi)=\mathcal{L}(\mathcal A_\varphi)$.
\end{theorem}

Such $\mathcal{A}_\varphi$ in Theorem \ref{DFA} can be constructed explicitly or symbolically using existing tools, such as SPOT \cite{duret2016spot} and MONA \cite{henriksen1995mona}.

\begin{remark}
For a given LTL$_F$ formula $\varphi$ over atomic propositions $\Pi$, the associated DFA $\mathcal A_\varphi$ is usually constructed over the alphabet $\Sigma = 2^\Pi$.
Solution process of a system $\mathfrak S$ is also connected to the set of words by a labeling function $L$ from the state set to the alphabet $\Sigma$. Without loss of generality, we work with the set of atomic propositions directly as the alphabet rather than its power set.
\end{remark}
\subsection{Property satisfaction by stochastic control systems}
For a given \Sys $\mathfrak S=(X,V_{\textsf w},U,w,f)$ with dynamics \eqref{DTSCS}, the system $\mathfrak S$ is connected to LTL$_F$ formulas with the help of a measurable labeling function $L: X \rightarrow \Pi$, where $\Pi$ is the set of atomic propositions.

\begin{definition}\label{sys_trace}
Consider a finite state sequence $\textbf{x}_N=(x(0),x(1),\ldots,x(N-1))\in X^N$, $N\in\N$, and labeling function $L: X \rightarrow \Pi$. Then, the corresponding trace is given by $L(\textbf{x}_N):=(\sigma_0,\sigma_1,\ldots,\sigma_{N-1}) \in\Pi^N$ if we have $\sigma_k=L(x(k))$ for all $k \in\{0,1,\ldots,N-1\}$.
 \end{definition}
 Note that we abuse the notation by using map $L(\cdot)$ over the domain $X^N$, i.e. $L(x(0),x(1),\ldots,x(N-1))\equiv (L(x(0)),L(x(1)),\ldots,L(x(N-1)))$.
 Their distinction is clear from the context.
 Next, we define the probability that a \Sys $\mathfrak S$ satisfies LTL$_F$ formula $\varphi$ over traces of length $N$.
 \begin{definition}
Consider a \Sys $\mathfrak S=(X,V_{\textsf w},U,w,f)$ and a LTL$_F$ formula $\varphi$ over $\Pi$. We denote by $\mathbb{P}^{x_0}_{\rho}\{L(\textbf{x}_N) \models \varphi\}$ the probability that $\varphi$ is satisfied by the state evolution of the system $\mathfrak S$ over a finite-time horizon $[0,N)\subset\N$ starting from initial state $x(0)=x_0\in X$ under control policy $\rho$. 
 \end{definition}

\begin{remark}
The set of atomic propositions $\Pi=\{p_0,p_1,\ldots,p_M\}$ and the labeling function $L: X \rightarrow \Pi$ provide a measurable partition of the state set $X = \cup_{i=1}^M X_i$ as  $X_i:=L^{-1}(p_i)$. We assume that $X_i\neq \emptyset$ for any $i$. This assumption is without loss of generality since all the atomic propositions $p_i$ with $L^{-1}(p_i) = \emptyset$ can be replaced by $(\neg\textsf{true})$ without affecting the probability of satisfaction.
\end{remark}

\subsection{Problem formulation}
\begin{problem}\label{SCS_prob}
Given a \Sys $\mathfrak S = (X,V_{\textsf w},U,w,f)$ with dynamics \eqref{DTSCS}, a LTL$_F$ specification $\varphi$ of length $N$ over a set of atomic propositions $\Pi=\{p_0,p_1,\ldots,p_M\}$, a labeling function $L: X \rightarrow \Pi$,
and real value $\vartheta\in(0,1)$,
 compute a control policy $\rho$ (if existing) such that $\mathbb{P}^{x_0}_\rho\{L(\textbf{x}_N) \models \varphi\}\ge \vartheta$ for all $x_0\in L^{-1}(p_i)$ and some $i\in\{1,2,\ldots,M\}$.
\end{problem}


Finding a solution to Problem~\ref{SCS_prob} (if existing) is difficult in general. In this paper, we give a computational method that is sound in solving the problem. Our approach is to compute a policy $\rho$ together with a lower bound $\underline{\vartheta}$. We try to find the largest lower bound, which then can be compared with $\vartheta$ and gives $\rho$ as a solution for Problem~\ref{SCS_prob} if $\underline{\vartheta}\ge \vartheta$.
{To solve this problem, we utilize the notion of control barrier certificates (discussed in Section~\ref{Sec_CBC}). In general, this notion is useful for providing an upper bound on the reachability probability. The negation of LTL$_F$ properties can then be equivalently represented as a sequence of reachability problems using a DFA. Therefore, instead of computing a control policy that guarantees a lower bound $\underline{\vartheta}$ on the probability satisfaction of the LTL$_F$ specification,
	we compute a policy that guarantees an upper bound on the probability satisfaction of its negation, \ie, $\mathbb{P}^{x_0}_\rho\{L(\textbf{x}_N) \models\neg \varphi\}\le \overline{\vartheta}$ for any $x_0\in L^{-1}(p_i)$ and some $i\in\{0,1,\ldots,M\}$.} Then for the same control policy the lower bound can be easily obtained as $\underline{\vartheta}=1-\overline{\vartheta}$.  
This is done by constructing a DFA $\mathcal{A}_{\neg \varphi}=(Q,Q_0,\Pi,\delta,F)$ that accepts all finite words over $\Pi$ satisfying $\neg\varphi$.

For the sake of illustrating the results better, we provide the following running example throughout the paper.


\def\example{\par\noindent{\textbf{Example 1.}} \ignorespaces}
\def\endexample{}

\begin{example}
\begin{figure}[b] \label{SCS_fig1}
			\centering
			\subfigure[]{\includegraphics[scale=0.6, height = 5cm]{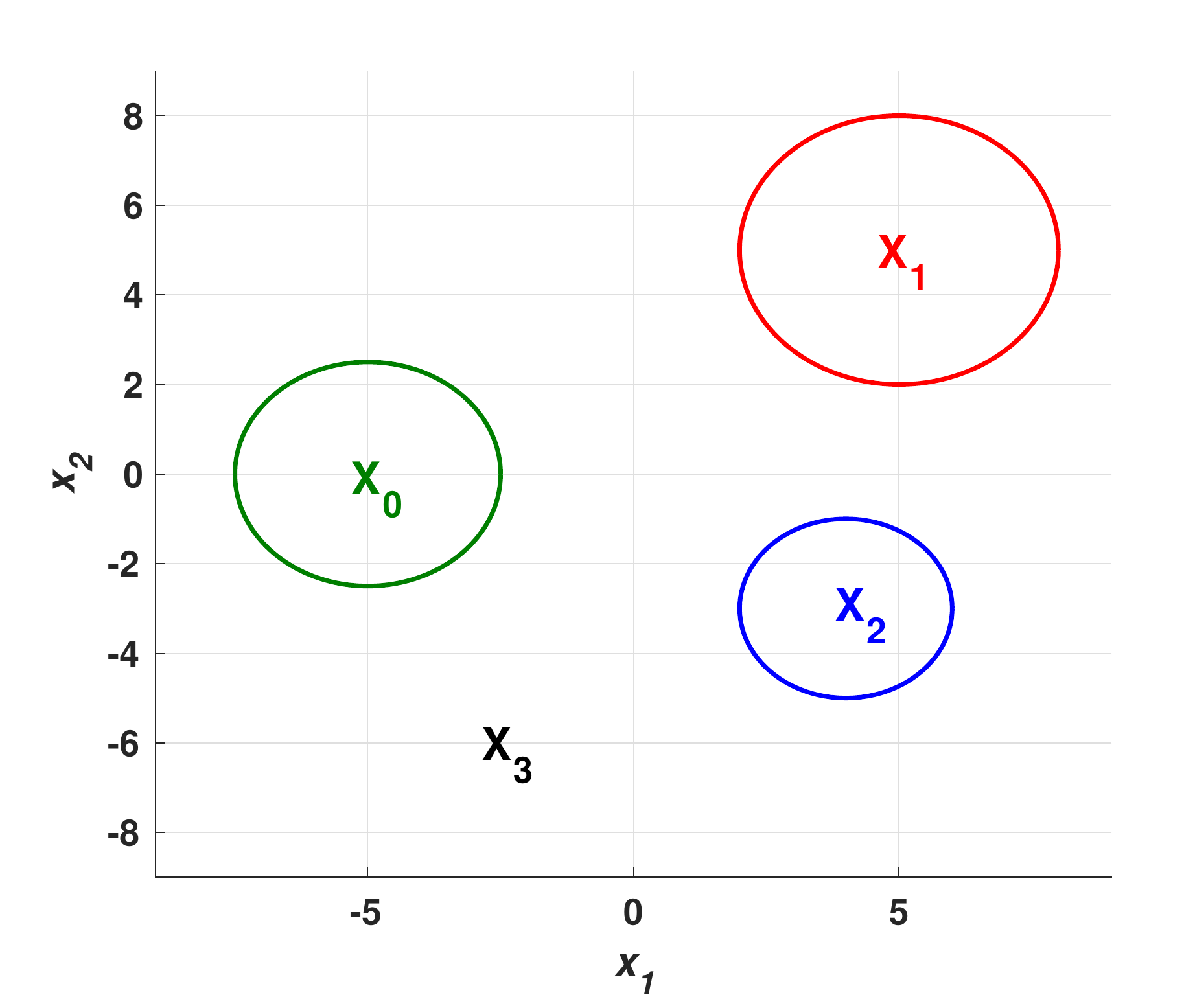}}  \hspace{.5em}
			\subfigure[]{\includegraphics[scale=0.5, height = 5cm]{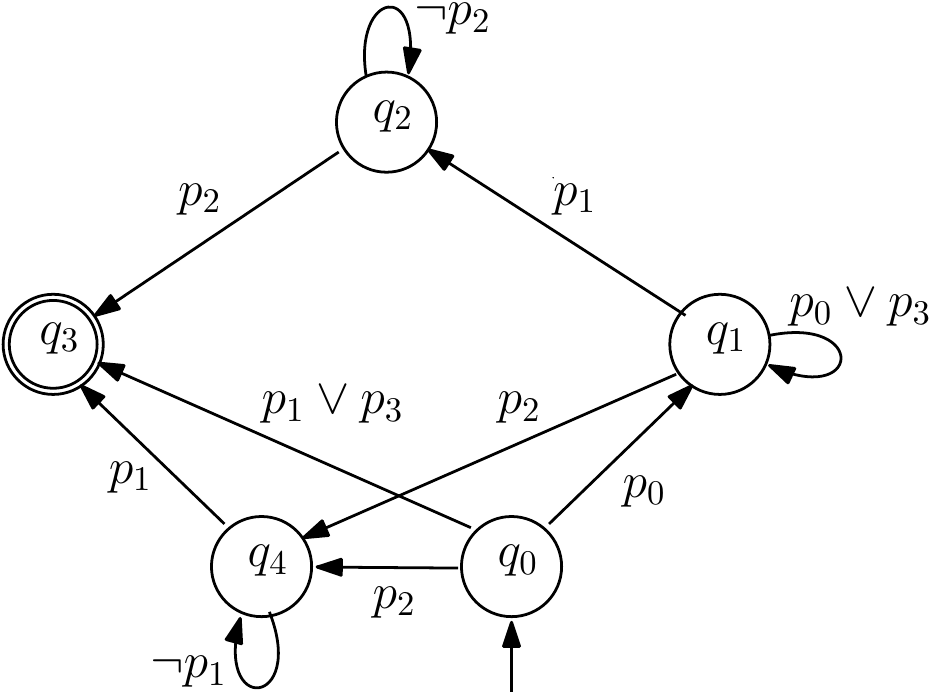}}\\
			\caption{(a) State set and regions of interest for Example~1, (b) DFA $\mathcal{A}_{\neg \varphi}$ that accepts all traces satisfying $\neg \varphi$ where $\varphi$ is given in \eqref{LTL_f}.}
\end{figure}
Consider a two-dimensional \Sys $\mathfrak S = (X,V_{\textsf w},U,w,f)$ with $X = V_{\textsf w} = \mathbb R^2$, $U = \mathbb R$ and dynamics
\begin{align}
x_1(k+1)&=x_1(k)-0.01 x_2^2(k)+0.5 w_1(k),\nonumber\\
x_2(k+1)&=-0.01 x_1(k) x_2(k)+u(k)+0.5 w_2(k),
\end{align}
where $u(\cdot)$ is a control input and $w_1(k)$, $w_2(k)$ are standard normal random variables that are independent from each other and for any $k\in \mathbb N_0$.
The set of atomic propositions is given by $\Pi=\{p_0,p_1,p_2,p_3\}$, with labeling function $L(x) = p_i$ for any $x\in X_i$, $i\in\{0,1,2,3\}$.
The sets $X_i$ are defined as 
\begin{align*}
X_0&=\{(x_1,x_2)\in X \mid (x_1+5)^2+x_2^2\leq 2.5\},\\
X_1&=\{(x_1,x_2)\in X \mid (x_1-5)^2+(x_2-5)^2\leq 3 \},\\
X_2&=\{(x_1,x_2)\in X \mid (x_1-4)^2+(x_2+3)^2\leq 2 \}, \text{ and }\\
X_3& = X\setminus (X_0\cup X_1\cup X_2).
\end{align*}
These sets are shown in Figure 1(a).
We are interested in computing a control policy $\rho$ that provides a lower bound on the probability that the trajectories of $\mathfrak S$ of length $N$ satisfies the following specification: 
\begin{itemize}
\item 
If it starts in $X_0$, it will always stay away from $X_1$ or always stay away from $X_2$. If it starts in $X_2$, it will always stay away from $X_1$.
\end{itemize}
This property can be expressed by the LTL$_F$ formula
\begin{equation}
\label{LTL_f}
\varphi=(p_0\wedge(\square \neg p_1 \vee \square\neg p_2))\vee (p_2\wedge\square\neg p_1).
\end{equation}
The DFA corresponding to the negation of $\varphi$ in \eqref{LTL_f} is shown in Figure 1(b).
\end{example}


\section{Control Barrier Certificates}\label{Sec_CBC}
In this section, we introduce the notion of control barrier certificate which will later serve as the core element for solving Problem~\ref{SCS_prob}. {Intuitively, control barrier certificates are relaxed versions of supermartingales that are decreasing in expectation along the trajectories of the system up to a constant. Once a barrier certificate is found while satisfying some conditions, it can give upper bounds on the reachability probability of system trajectories.} 
\begin{definition}
\label{CBC_def}
A function $B:X\rightarrow\mathbb{R}_0^+$ is a control barrier certificate for a \Sys $\mathfrak S=(X,V_{\textsf w},U,w,f)$ if for any state $x\in X$, there exists an input $u\in U$ such that
\begin{align}
\label{cmart}
 \mathbb{E}[B(f(x,u,w))\mid x,u]\leq B(x)+c,
\end{align}
 for some constant $c\geq 0$.
\end{definition}
If the set of control inputs $U$ is finite, one can rewrite Definition~\ref{CBC_def} as follows.
\begin{definition}
\label{CBC_def1}
A function $B:X\rightarrow\mathbb{R}_0^+$ is a control barrier certificate for a \Sys $\mathfrak S=(X,V_{\textsf w},U,w,f)$ with $U=\{u_1,u_2,\ldots,u_l\}$, $l\in\N$, if 
\begin{align}\label{cmart1}
\min_{u\in U} \mathbb{E}[B(f(x,u,w))\mid x,u]\leq B(x)+c \quad \forall x\in X,
\end{align}
for some constant $c\geq 0$.
\end{definition}
{\begin{remark}\label{key1}
		Note that  conditions~\eqref{cmart}-\eqref{cmart1} are relaxed versions of so-called supermartingale condition. This is due to the positive constant $c$ on the right-hand side. When $c=0$, the function $B(\cdot)$ becomes supermartingale for $\mathfrak S$.  
\end{remark}}
\begin{remark} \label{control_policy}
The above definitions associate a stationary policy $u:X\rightarrow U$ to a control barrier certificate. Definition~\ref{CBC_def} gives such a policy according to the existential quantifier on the input for any state $x\in X$. Definition~\ref{CBC_def1} gives the policy as the $\argmin$ of the left-hand side of inequality \eqref{cmart1}.
In case of discrete inputs, $u(x)$ can be selected as an element of $\{u \in U\mid \mathbb{E}[B(f(x,u,w))\mid x,u]\leq B(x)+c\}$. In other words, Definition~\ref{CBC_def1} provides regions of state-set in which the particular control input is valid and is given as $\mathsf X_i:=\{x\in X\mid \mathbb{E}[B(f(x,u_i,w))\mid x,u_i]\leq B(x)+c\}$ for all $i\in\{1,2,\ldots,l\}$ and $\bigcup_{i}\mathsf X_i=X$.
\end{remark}

We provide the following lemma and use it in the sequel. This lemma is a direct consequence of \cite[Theorem 3]{kushnerbook} and is also utilized in \cite[Theorem II.1]{steinhardt2012finite}.

\begin{lemma}\label{CBC_lem1}
Consider a \Sys $\mathfrak S=(X,V_{\textsf w},U,w,f)$ and let $B: X\rightarrow \mathbb R_0^{+}$ be a control barrier certificate as given in Definition~\ref{CBC_def} (or Definition~\ref{CBC_def1}) with constant $c$ and stationary policy $u:X\rightarrow U$ {as discussed in Remark \ref{control_policy}}. Then for any constant $\lambda>0$ and any initial state $x_0\in X$, 
\begin{align}
\mathbb{P}^{x_0}_{u}\{\sup_{0\leq k < T_d}B(x(k))\geq\lambda\mid x(0)=x_0\}\leq\frac{B(x_0)+cT_d}{\lambda}.
\label{eq_CBC_lem1}
\end{align}
\end{lemma}
\begin{proof}
The proof is similar to that of Theorem 3 in \cite{kushnerbook} and is omitted here.
\end{proof}
Next theorem shows that a control barrier certificate can give an upper bound on the probability of satisfying reachability specification. This theorem is inspired by the result of \cite[Theorem 15]{4287147} that uses supermartingales for reachability analysis of continuous-time stochastic systems.


\begin{theorem}
\label{barrier}
Consider a \Sys $\mathfrak S=(X,V_{\textsf w},U,w,f)$ and sets $X_a, X_b\subseteq X$. Suppose there exist a control barrier certificate $B:X\rightarrow\mathbb{R}_0^+$ as defined in Definition~\ref{CBC_def} (or Definition~\ref{CBC_def1}) with constant $c\ge 0$ and stationary policy $u:X\rightarrow U$ {as discussed in Remark \ref{control_policy}}. If there is a constant $\gamma\in[0,1]$ such that
\begin{align}
&B(x)\leq\gamma &\forall x\in X_a,\label{ine_1}\\
&B(x)\geq 1 &\forall x\in X_b\label{ine_2},
\end{align}
then the probability that the state evolution of $\mathfrak S$ starts from any initial state $x_0\in X_a$ and reaches $X_b$ under policy $u(\cdot)$ within time horizon $[0,T_d)\subseteq \N_0$ is upper bounded by $\gamma+cT_d$.
\end{theorem}
\begin{proof}
Since $B(x(k))$ is a control barrier certificate, we conclude that \eqref{eq_CBC_lem1} in Lemma~\ref{CBC_lem1} holds. Now using \eqref{ine_1} and the fact that $X_b\subseteq \{x\in X \mid B(x) \geq 1 \}$, we have $\mathbb{P}_{u}^{x_0}\{x(k)\in X_1 \text{ for some } 0\leq k< T_d\mid x(0)=x_0\}$ $\leq\mathbb{P}_{u}^{x_0}\{ \sup_{0\leq k< T_d} B(x(k))\geq 1\mid x(0)=x_0 \}$ $\leq B(x_0)+cT_d$ $\leq \gamma+cT_d$, which concludes the proof.
\end{proof}

Theorem~\ref{barrier} enables us to formulate an optimization problem for finding a sound solution of the policy synthesis problem \ref{SCS_prob} with reachability specifications. We can minimize the values of $\gamma$ and $c$ in order to find an upper bound for finite-horizon reachability that is as tight as possible. 
{\begin{remark}
		If one succeeds in finding a control barrier certificate $B(\cdot)$ with $c=0$ satisfying conditions of Theorem~\ref{barrier}, the result of the theorem holds for an unbounded time horizon. However, considering relaxed supermartingale condition as discussed in Remark \ref{key1}, makes it easier to find $B(\cdot)$ satisfying conditions in Theorem \ref{barrier} and makes out results applicable to larger classes of systems.  
\end{remark}}


\section{Decomposition into Sequential Reachability}
\label{runs}
In this section, we discuss how to translate the synthesis problem~\ref{SCS_prob} for any LTL$_F$ specification into a sequence of simple reachability tasks that can be solvable by computing control barrier certificates as discussed in Theorem~\ref{barrier}.
Consider a DFA $\mathcal{A}_{\neg \varphi}=(Q,Q_0,\Pi,\delta,F)$ that accepts all finite words of length $n\in[0,N]\subset\N_0$ satisfying $\neg \varphi$.

\noindent\textbf{Accepting state run of $\mathcal{A}_{\neg \varphi}$.} For any $n\in\N_0$, sequence $\textbf{q}=(q_0,q_1,\ldots,q_n)\in Q^{n+1}$ is called an accepting state run if $q_0\in Q_0$, $q_n\in F$, and there exist a finite word $\sigma = (\sigma_0,\sigma_1,\ldots,\sigma_{n-1})\in\Pi^n$ such that $q_i \overset{\sigma_i}{\longrightarrow} q_{i+1}$ for all $i\in\{0,1,\ldots, n-1\}$.
We denote the set of such finite words by $\sigma(\textbf{q})\subseteq \Pi^n$ and the set of accepting state runs by $\mathcal R$. We also indicate the length of $\textbf{q}\in Q^{n+1}$ by $|\textbf{q}|$, which is $n+1$.

Self-loops in the DFA play a central role in our decomposition.
Let $Q_s\subseteq Q$ be a set of states of $\mathcal{A}_{\neg \varphi}$ having self-loops, \ie, $Q_s:= \{q\in Q \,|\, \exists p\in \Pi, q\overset{p}{\longrightarrow} q\}$.
Let $\mathcal{R}_{N}$ be the set of all finite accepting state runs of lengths less than or equal to $N+1$ excluding self-loops,
\begin{equation}
\label{eq:runs}
\mathcal{R}_{N} := \{\textbf{q} = (q_0,q_1,\ldots,q_n)\in\mathcal R \,|\, n\le N,\, q_i\neq q_{i+1},\, \forall i<n\}.
\end{equation}
Computation of $\mathcal{R}_{N}$ can be done efficiently using algorithms in graph theory by viewing $\mathcal{A}_{\neg \varphi}$ as a directed graph.
Consider $\mathcal{G}=(\mathcal{V},\mathcal{E})$ as a directed graph with vertices $\mathcal{V}=Q$ and edges $\mathcal{E}\subseteq\mathcal{V}\times\mathcal{V}$ such that $(q,q')\in\mathcal{E}$ if and only if $q'\neq q$ and there exist $p\in\Pi$ such that $q\overset{p}{\longrightarrow} q'$. For any $(q,q')\in\mathcal{E}$, we denote the atomic proposition associated with the edge $(q,q')$ by $\sigma(q,q')$. From the construction of the graph, it is obvious that the finite path in the graph of length $n+1$ starting from vertices $q_0\in Q_0$ and ending at $q_F\in F$ is an accepting state run $\textbf{q}$ of $\mathcal{A}_{\neg \varphi}$ without any self-loop thus belongs to $\mathcal{R}_{N}$.
Then one can easily compute $\mathcal{R}_{N}$ using variants of depth first search algorithm \cite{russell2003artificial}. For each $p\in\Pi$, we define a set  $\mathcal{R}^p_{N}$ as
\begin{equation}\label{eq:run1}
\mathcal{R}^p_{N}:=\{\textbf{q}=(q_0,q_1,\ldots,q_n)\in\mathcal{R}_{N}\mid \sigma(q_0,q_1)=p\in\Pi \}.
\end{equation}
Note that we use the superscript $p\in\Pi$ to represent the atomic proposition corresponding to the initial region from which the state evolution starts. We use a similar notation throughout the paper.\\
Decomposition into sequential reachability is performed as follows.
For any $\textbf{q} = (q_0,q_1,\ldots,q_n)\in\mathcal{R}^p_{N}$, we define $\mathcal{P}^p(\textbf{q})$ as a set of all state runs of length $3$ augmented with a horizon,
\begin{equation}
\label{eq:reachability}
\mathcal{P}^p(\textbf{q}):=\{\left(q_i,q_{i+1},q_{i+2},T(\textbf{q},q_{i+1})\right)\mid 0\leq i\leq n-2\},
\end{equation}
where the horizon is defined as $T(\textbf{q},q_{i+1}) = N+2-|\textbf{q}|$ for $q_{i+1}\in Q_s$ and $1$ otherwise. Note that the state runs of length $3$ in \eqref{eq:reachability} corresponds to two atomic propositions associated with respective edges which will later serve as regions $X_a$ and $X_b$ and the term $T(\textbf{q},q_{i+1})$ in \eqref{eq:reachability} will serve as $T_d$ in Theorem \ref{barrier}. We denote $\mathcal{P}(\mathcal{A}_{\neg \varphi})=\bigcup_{p\in\Pi}\bigcup_{\textbf{q}\in\mathcal{R}^p_{N}}\mathcal{P}^p(\textbf{q})$.
\begin{remark}
Note that $\mathcal{P}^p(\textbf{q})=\emptyset$ for $|\textbf{q}|=2$. In fact, any accepting state run of length $2$ specifies a subset of the state set such that the system satisfies $\neg\varphi$ whenever it starts from that subset. This gives trivial zero probability for satisfying the specification, thus neglected in the sequel. 
\end{remark}
The computation of sets $\mathcal{P}^p(\textbf{q})$, $\textbf{q}\in \mathcal R^p_{N}$, $p\in\Pi$, is illustrated in Algorithm~\ref{algo1} and demonstrated below for our running example.


\begin{algorithm}[t]
	\caption{Computation of sets $\mathcal{P}^p(\textbf{q})$, $\textbf{q}\in\mathcal{R}^p_{N}$, $p\in\Pi$}
	\label{algo1}
	\begin{algorithmic}[1]
		\Require{$ \mathcal{G} $, $ Q_s $, $N$, $\Pi$}
		\Initialize{$\mathcal{P}^p(\textbf{q})\leftarrow \emptyset,\quad\forall p\in\Pi$}
		\State Compute set $\mathcal{R}_{N}$ by depth first search on $ \mathcal{G} $
		\ForAll{$ \textbf{q}=(q_0,q_1,\ldots,q_n)\in \mathcal{R}_{N} $ and $p\in\Pi$}
		\If {$p=\sigma(q_0,q_1)$}
		\State{$\mathcal{R}^p_{N}\leftarrow\{\textbf{q}\}$}
		\EndIf
		\EndFor
		\ForAll{$p\in\Pi$ and $ \textbf{q}\in \mathcal{R}^p_{N} $ and $|\textbf{q}|\geq3$}
		\For {$i=0$ to $|\textbf{q}|-3$}
		\State $\mathcal{P}_{temp}(\textbf{q})\leftarrow \{(q_i,q_{i+1},q_{i+2})\}$
		\If {$q_{i+1}\in Q_s$}
		\State $\mathcal{P}^p(\textbf{q})\leftarrow \{(q_i,q_{i+1},q_{i+2},N+2-|\textbf{q}|)\}$
		\Else
		\State $\mathcal{P}^p(\textbf{q})\leftarrow \{(q_i,q_{i+1},q_{i+2},1)\}$
		\EndIf
		\EndFor
		\EndFor
		\Return $\mathcal{P}^p(\textbf{q}),\quad\forall p\in\Pi$
	\end{algorithmic}
\end{algorithm}


\def\example{\par\noindent{\bf Example 1.} \ignorespaces}
\def\endexaple{}

\begin{example}(continued)
For LTL$_F$ formula $\varphi$ given in \eqref{LTL_f}, Figure 1(b) shows a DFA $\mathcal{A}_{\neg\varphi}$ that accepts all words that satisfy $\neg\varphi$. From Figure 1(b), we get $Q_0=\{q_0\}$, $\Pi=\{p_0,p_1,p_2,p_3\}$ and $F=\{q_3\}$. We consider traces of maximum length $N=5$.
The set of accepting state runs of lengths at most $N+1$ without self-loops is
 \begin{equation*}
 \mathcal{R}_{5}=\{(q_0,q_4,q_3),(q_0,q_1,q_2,q_3),(q_0,q_1,q_4,q_3),(q_0,q_3)\}.
 \end{equation*}
The sets $\mathcal{R}^p_{5}$ for $p\in\Pi$ are as follows:
\begin{align*}
&\mathcal{R}^{p_0}_{5}=\{(q_0,q_1,q_2,q_3),(q_0,q_1,q_4,q_3)\},\quad\mathcal{R}^{p_1}_{5}=\{(q_0,q_3)\}, \quad \mathcal{R}^{p_2}_{5}=\{(q_0,q_4,q_3)\},\quad \mathcal{R}^{p_3}_{5}=\{(q_0,q_3)\}.
\end{align*} 
The set of states with self-loops is $Q_s = \{q_1,q_2,q_4\}$. Then the sets $\mathcal{P}^p(\textbf{q})$ for $\textbf{q}\in\mathcal{R}^p_{5}$ are as follows:
\begin{align*}
&\mathcal{P}^{p_0}(q_0,q_1,q_2,q_3)=\{(q_0,q_1,q_2,3),(q_1,q_2,q_3,3)\},\\
&\mathcal{P}^{p_0}(q_0,q_1,q_4,q_3)=\{(q_0,q_1,q_4,3),(q_1,q_4,q_3,3)\},\\
&\mathcal{P}^{p_1}(q_0,q_3)=\mathcal{P}^{p_3}(q_0,q_3)=\emptyset,\quad \mathcal{P}^{p_2}(q_0,q_4,q_3)=\{(q_0,q_4,q_3,4)\}.
\end{align*} 
For every $\textbf{q}\in\mathcal{R}^p_{5}$, the corresponding finite words $\sigma(\textbf{q})$ are listed as follows:
\begin{align*}
& \sigma(q_0,q_3)=\{p_1\},\quad \sigma(q_0,q_4,q_3)=\{(p_2,p_1)\},\\
& \sigma(q_0,q_1,q_2,q_3)=\{(p_0,p_1,p_2)\},\quad \sigma(q_0,q_1,q_4,q_3)=\{(p_0,p_2,p_1)\}.
\end{align*}
\end{example}


\section{Controller Synthesis using Control Barrier Certificates}
\label{BC_computation}
Having $\mathcal{P}^p(\textbf{q})$ defined in \eqref{eq:reachability} as the set of state runs of length $3$ augmented with a horizon, in this section, we provide a systematic approach to compute a policy with a (potentially tight) lower bound on the probability that the state evolutions of $\mathfrak S$ satisfies $\varphi$. Given DFA $\mathcal{A}_{\neg\varphi}$, our approach relies on performing a reachability computation over each element of $\mathcal P(\mathcal{A}_{\neg\varphi})$ {(\ie, $\bigcup_{p\in\Pi}\bigcup_{\textbf{q}\in\mathcal{R}^p_{N}}\mathcal{P}^p(\textbf{q})$)}, where reachability probability is upper bounded using control barrier certificates along with appropriate choices of control inputs as mentioned in Theorem~\ref{barrier}. However, computation of control barrier certificates and the policies for each element $\nu\in\mathcal P(\mathcal{A}_{\neg\varphi})$, can cause ambiguity while utilizing controllers in closed-loop whenever there are more than one outgoing edges from a state of the automaton. To make it more clear, consider elements $\nu_1=(q_0,q_1,q_2, T((q_0,q_1,q_2,q_3),q_1))$ and $\nu_2=(q_0,q_1,q_4, T((q_0,q_1,q_4,q_3),q_1))$ from Example 1, where there are two outgoing transitions from state $q_1$ (see Figure 1(b)). This results in two different reachability problems, namely, reaching sets $L^{-1}(\sigma(q_1,q_2))$ and $L^{-1}(\sigma(q_1,q_4))$ starting from the same set $L^{-1}(\sigma(q_0,q_1))$. Thus computing different control barrier certificates and corresponding controllers in such a scenario is not helpful. To resolve this ambiguity, we simply merge such reachability problems into one reachability problem by replacing the reachable set $X_b$ in Theorem~\ref{barrier} with the union of regions corresponding to the alphabets of all outgoing edges. Thus we get a common control barrier certificate and a corresponding controller. This enables us to partition $\mathcal P(\mathcal{A}_{\neg\varphi})$ and put the elements sharing a common control barrier certificate and a corresponding control policy in the same partition set. These sets can be formally defined as
$$\mu_{(q,q',\Delta(q'))}:=\{(q,q',q'',T)\in\mathcal P(\mathcal{A}_{\neg\varphi})\mid q,q',q''\in Q\text{ and }q''\in\Delta(q')\}.$$
The control barrier certificate and the control policy corresponding to the partition set $\mu_{(q,q',\Delta(q'))}$ are denoted by $B_{\mu_{(q,q',\Delta(q'))}}(x)$ and $u_{\mu_{(q,q',\Delta(q'))}}(x)$, respectively. Thus, for all $\nu\in\mathcal P(\mathcal{A}_{\neg\varphi})$, we have 
\begin{equation}\label{eq_SCH_controller}
B_\nu(x)=B_{\mu_{(q,q',\Delta(q'))}}(x)\text{ and } u_\nu(x)=u_{\mu_{(q,q',\Delta(q'))}}(x),\quad \text{if }\nu\in\mu_{(q,q',\Delta(q'))}.
\end{equation}

 \subsection{Control policy}\label{aaaaa1}
From the above discussion, one can readily observe that we have different control policies at different locations of the automaton which can be interpreted as a switching control policy.
Next, we define the automaton representing the switching mechanism for control policies. Consider the DFA $\mathcal{A}_{\neg \varphi}=(Q,Q_0,\Pi,\delta,F)$ corresponding to $\neg\varphi$ as discussed in Section~\ref{runs}, {where $\Delta(q)$ denotes the set of all successor states of $q\in Q$.} Now, the switching mechanism is given by a DFA $\mathcal{A}_{\mathfrak m}=(Q_{\mathfrak m},Q_{\mathfrak m 0},\Pi_{\mathfrak m},\delta_{\mathfrak m},F_{\mathfrak m})$, where $Q_{\mathfrak m}:=Q_{\mathfrak m 0}\cup\{(q,q',\Delta(q'))\mid q,q'\in Q\setminus F\}\cup F_{\mathfrak m}$ is the set of states, $Q_{\mathfrak m 0}:=\{(q_0,\Delta(q_0))\mid q_0\in Q_0\}$ is a set of initial states, $\Pi_{\mathfrak m}=\Pi$, $F_{\mathfrak m}=F$, and the transition relation $(q_{\mathfrak m},\sigma,q_{\mathfrak m}')\in \delta_{\mathfrak m}$ is defined as
\begin{itemize}
\item for all $q_{\mathfrak m}=(q_0,\Delta(q_0))\in Q_{\mathfrak m 0}$, 
\begin{itemize}
\item $(q_0,\Delta(q_0))\overset{\sigma(q_0,q'')}{\longrightarrow}(q_0,q'',\Delta(q''))$, where $q_0\!\!\overset{\sigma(q_0,q'')}{\longrightarrow}q''$;
\end{itemize}
\item for all $q_{\mathfrak m}=(q,q',\Delta(q'))\in Q_{\mathfrak m}\setminus (Q_{\mathfrak m 0}\cup F_{\mathfrak m})$,
\begin{itemize}
 \item $(q,q',\Delta(q'))\overset{\sigma(q',q'')} {\longrightarrow}(q',q'',\Delta(q''))$, such that $q,q',q''\in Q$, $q''\in \Delta(q')$ and $q''\notin F$; and
 \item $(q,q',\Delta(q'))\overset{\sigma(q',q'')} {\longrightarrow} q''$, such that $q,q',q''\in Q$, $q''\in \Delta(q')$ and $q''\in F$.
 \end{itemize}
\end{itemize}
The control policy that is a candidate for solving Problem~\ref{SCS_prob} is given as 
\begin{equation}\label{eq:policy}
\rho(x,q_{\mathfrak m})=u_{\mu_{(q_{\mathfrak m}')}}(x), \quad \forall (q_{\mathfrak m},L(x),q_{\mathfrak m}')\in\delta_{\mathfrak m}.
\end{equation}
In the next subsection, we discuss the computation of bound on the probability of satisfying the specification under such a policy, which then can be used for checking if this policy is indeed a solution for Problem~\ref{SCS_prob}.
\begin{remark}
The control policy in \eqref{eq:policy} is a Markov policy on the augmented space $X\times Q_{\mathfrak m}$. Such a policy is equivalent to a history dependent policy on the state set $X$ of the system as discussed in Subsection~\ref{Sec_dtSCS} (see \cite{tkachev2013quantitative} for a proof).  
\end{remark}
\begin{example}(continued)
The DFA $\mathcal{A}_{\mathfrak m}=(Q_{\mathfrak m},Q_{\mathfrak m 0},\Pi_{\mathfrak m},\delta_{\mathfrak m},F_{\mathfrak m})$ modeling the switching mechanism between policies for the system in Example 1 is shown in Figure~\ref{SCH_swithing}. 
\begin{figure}[h] 
			\centering
			\includegraphics[scale=0.6, height = 5.5cm]{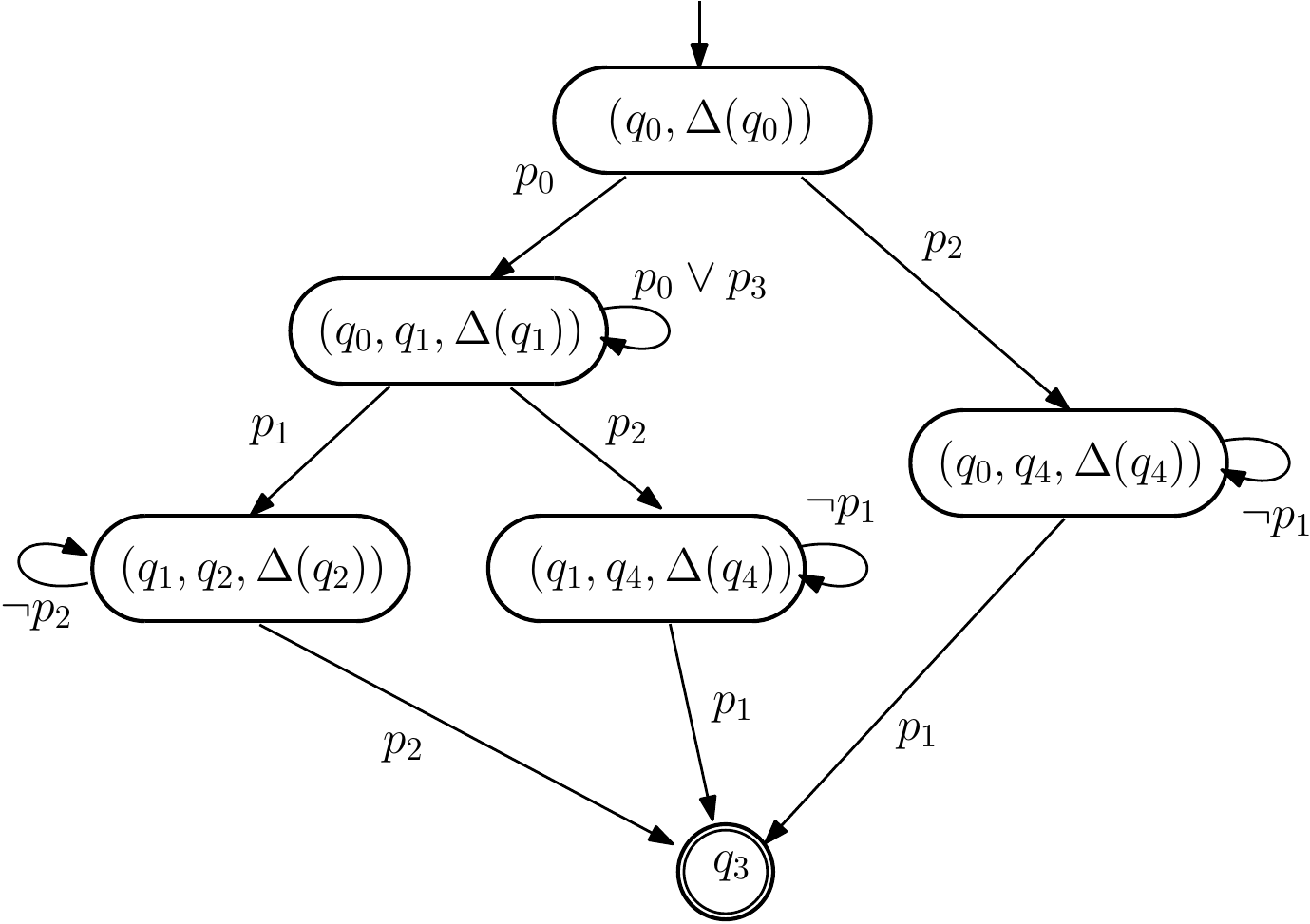} 
			\caption{DFA $\mathcal{A}_{\mathfrak m}$ representing switching mechanism for controllers for Example 1.}
			\label{SCH_swithing}
\end{figure}
 \qed
\end{example}



\subsection{Computation of probabilities}\label{aaaaa2}
Next theorem provides an upper bound on the probability that the state evolution of the system satisfies the specification $\neg\varphi$.


\begin{theorem}\label{upper_bound}
For a given LTL$_F$ specification $\varphi$, let $\mathcal{A}_{\neg\varphi}$ be a DFA corresponding to its negation. For $p\in\Pi$,  let $\mathcal R^p_{N}$ be the set defined in \eqref{eq:run1}, and $\mathcal{P}^p$ be the set of runs of length $3$ augmented with a horizon defined in \eqref{eq:reachability}.  The probability that the state evolution of $\mathfrak S$ starting from any initial state $x_0\in L^{-1}(p)$ under the control policy in \eqref{eq:policy} 
satisfies $\neg\varphi$ within time horizon $[0,N)\subseteq\mathbb{N}_0$ is upper bounded by
\begin{equation}
\label{eq:bound}
\mathbb{P}^{x_0}_\rho\{L(\textbf{x}_N) \models\neg \varphi\}\leq \sum_{\textbf{q} \in\mathcal R^p_{N}}\prod\left\{(\gamma_\nu+c_\nu T)\,|\, \nu = (q,q',q'',T)\in \mathcal{P}^p(\textbf{q})\right\}, 
\end{equation}
where $\gamma_\nu+c_\nu T$ is computed via Theorem~\ref{barrier} which is the upper bound on the probability of the trajectories of $\mathfrak S$ starting from $X_a := L^{-1}(\sigma(q,q'))$ and reaching $ X_b:=L^{-1}(\sigma(q',q''))$ within time horizon $[0,T)\subseteq\mathbb{N}_0$.
\end{theorem}

\begin{proof}
For $p\in\Pi$, consider an accepting run $\textbf{q}\in \mathcal{R}^p_{N}$ and set $\mathcal{P}^p(\textbf{q})$ as defined in \eqref{eq:reachability}. We apply Theorem~\ref{barrier} to any $\nu=(q,q',q'',T)\in \mathcal{P}^p(\textbf{q})$. The probability that the state evolution of $\mathfrak S$ starts from any initial state $x_0\in L^{-1}(\sigma(q,q'))$ and reaches $L^{-1}(\sigma(q',q''))$ under control input $u_\nu(x)$ within time horizon $[0,T]\subseteq \N_0$ is upper bounded by $\gamma_\nu+c_\nu T$. Now the upper bound on the probability of the trace of the state evolution (\ie, $L(\textbf{x}_N)$) reaching accepting state following trace corresponding to $\textbf{q}$ is given by the product of the probability bounds corresponding to all elements $\nu=(q,q',q'',T)\in \mathcal{P}^p(\textbf{q})$ and is given by
  \begin{equation}
\label{eq:bound_1}
\mathbb{P}\{ \sigma(\textbf{q})\models \neg \varphi \} \leq \prod \left\{(\gamma_\nu+c_\nu T)\,|\, \nu = (q,q',q'',T)\in \mathcal{P}^p(\textbf{q})\right\}.
\end{equation}
Note that, the way we computed time horizon $T$, we always get the upper bound for the probabilities for all possible combinations of self-loops for accepting state runs of length less than or equal to $N+1$.
The upper bound on the probability that the state evolution of the system $\mathfrak S$ starting from any initial state $x_0\in L^{-1}(p)$ violating $\varphi$ can be computed by summing the probability bounds for all possible accepting runs as computed in \eqref{eq:bound_1} and is given by
\begin{equation}
\label{eq:bound_2}
\mathbb{P}^{x_0}_\rho\{L(\textbf{x}_N) \models \neg \varphi\} \leq \sum_{\textbf{q} \in\mathcal R^p_{N}}\prod\left\{(\gamma_\nu+c_\nu T)\,|\, \nu = (q,q',q'',T)\in \mathcal{P}^p(\textbf{q})\right\}.\nonumber
\end{equation} 
\end{proof}
Theorem~\ref{upper_bound} enables us to decompose the computation into a collection of sequential reachability, compute bounds on the reachability probabilities using Theorem~\ref{barrier}, and then combine the bounds in a sum-product expression.
{Note that the upper bound provided in \eqref{eq:bound} could be replaced by $\min\big\{1,\sum_{\textbf{q} \in\mathcal R^p_{N}}\prod\{(\gamma_\nu+c_\nu T)\,|\, \nu \!=\! (q,q',q'',T)\!\in\! \mathcal{P}^p(\textbf{q})\}\big\}$,
	to prevent it from being greater than one. This bound is useful only if it is less than one.}
\begin{remark}
In case we are unable to find control barrier certificates for some of the elements $\nu \in \mathcal P^p(\textbf q)$ in \eqref{eq:bound}, we replace the related term $(\gamma_\nu+c_\nu T)$ by the pessimistic bound $1$. In order to get a non-trivial bound in \eqref{eq:bound}, at least one control barrier certificate must be found for each $\textbf{q} \in\mathcal R^p_{N}$.
\end{remark}
\begin{corollary}
\label{lower_bound}
Given the result of Theorem~\ref{upper_bound}, the probability that the trajectories of $\mathfrak S$ of length $N$  starting from any $x_0\in L^{-1}(p)$ satisfies LTL$_F$ specification $\varphi$ is lower-bounded by
$$\mathbb{P}^{x_0}_\rho\{L(\textbf{x}_N) \models \varphi\}\geq1-\mathbb{P}^{x_0}_\rho\{L(\textbf{x}_N) \models \neg \varphi\}.$$
\end{corollary}
\subsection{Computation of control barrier certificate}
\label{compute_CBC}
Proving the existence of a control barrier certificate and finding one are in general hard problems.
But if we restrict the class of systems and labeling functions, we can construct computationally efficient techniques to search for control barrier certificates and corresponding control policies of specific forms. In this subsection, we provide two possible approaches for computing control barrier certificates and corresponding control policies for a \Sys $\mathfrak S$ with respectively continuous and discrete input sets.
\subsubsection{Continuous input sets}\label{SOSS}
We propose a technique using sum-of-squares (SOS) optimization \cite{parrilo2003semidefinite}, relying on the fact that a polynomial is non-negative if it can be written as a sum of squares of different polynomials. In order to utilize an SOS optimization, we raise the following assumption.


\begin{assumption}
\label{ass:BC}
System $\mathfrak S$ has a continuous state set $X\subseteq \mathbb R^n$ and a continuous input set $U\subseteq \R^m$. Its vector field $f:X\times U\times V_{\textsf{w}}\rightarrow X$ is a polynomial function of state $x$ and input $u$ for any $w\in V_{\textsf{w}}$.
Partition sets $X_i = L^{-1}(p_i)$, $i\in\{0,1,2,\ldots, M\}$, are bounded semi-algebraic sets, \ie, they can be represented by polynomial equalities and inequalities.
\end{assumption}

Under Assumption~\ref{ass:BC}, we can formulate conditions in Theorem~\ref{barrier} as an SOS optimization to search for a polynomial control barrier certificate $B(\cdot)$, a polynomial control policy $u(\cdot)$ and a upper bound $(\gamma+c T_d)$. The following lemma provides a set of sufficient conditions for the existence of such control barrier certificate required in Theorem~\ref{barrier}, which can be solved as an SOS optimization.


\begin{lemma}
	\label{sos}
	Suppose Assumption~\ref{ass:BC} holds and sets $X_a,X_b, X$ can be defined by vectors of polynomial inequalities $X_a=\{x\in\R^n\mid g_0(x)\geq0\}$, $X_b=\{x\in\R^n\mid g_1(x)\geq0\}$, and $X=\{x\in\R^n\mid g(x)\geq0\}$, where the inequalities are defined element-wise.
	Suppose there exists a sum-of-square polynomial $B(x)$, constants $\gamma\in [0,1)$ and $c\ge 0$, polynomials $\lambda_{u_i}(x)$ corresponding to the $i^{\text{th}}$ input in $u=(u_1,u_2,\ldots,u_m)\in U\subseteq \R^m$, and vectors of sum-of-squares polynomials $\lambda_0(x)$, $\lambda_1(x)$, and {$\lambda_x(x,u)$} of appropriate size such that following expressions are sum-of-squares polynomials
	\begin{align}
	&\hspace{-0.4em}-B(x)-\lambda_0^T(x) g_0(x)+\gamma\label{eq:sos1}\\
	&B(x)-\lambda_1^T(x) g_1(x)-1\label{eq:sos2}\\
	&\hspace{-0.4em}-\hspace{-0.2em}\mathbb{E}[B(\hspace{-0.1em}f\hspace{-0.1em}(x,\hspace{-0.1em}u,\hspace{-0.1em}w\hspace{-0.1em})\hspace{-0.1em})|x,\hspace{-0.1em}u]\hspace{-0.2em}+\hspace{-0.2em}B(\hspace{-0.1em}x\hspace{-0.1em})\hspace{-0.2em}-\hspace{-0.3em}\sum_{i=1}^{m}(u_i\hspace{-0.2em}-\hspace{-0.3em}\lambda_{u_i}\hspace{-0.2em}(\hspace{-0.1em}x\hspace{-0.1em})\hspace{-0.1em})\hspace{-0.2em}-\hspace{-0.3em}\lambda_x^T\hspace{-0.2em}(x,\hspace{-0.1em}u) g(\hspace{-0.1em}x\hspace{-0.1em})\hspace{-0.2em}+\hspace{-0.2em}c.\label{eq:sos3}
	\end{align}
	Then,  $B(x)$ satisfies conditions in Theorem~\ref{barrier} and any {$u_i\geq\lambda_{u_i}(x)$} is the corresponding control input.
\end{lemma}
 \begin{proof}
 Since the entries $B(x)$ and $\lambda_0(x)$ in $-B(x)-\lambda_0^T(x) g_0(x)+\gamma$ are sum-of-squares, we have $0\leq B(x)+\lambda_0^T(x) g_0(x)\leq \gamma$. Since the term $\lambda_0^T(x) g_0(x)$ is non-negative over $X_a$, \eqref{eq:sos1} implies condition \eqref{ine_1} in Theorem~\ref{barrier}. Similarly, we can show that \eqref{eq:sos2} implies condition \eqref{ine_2} in Theorem~\ref{barrier}. Now consider \eqref{eq:sos3}. If we choose control input $u_i=\lambda_{u_i}(x)$ and since the term $\lambda^T(x) g(x)$ is non-negative over set $X$, we have $\mathbb{E}[B(f(x,u,w))|x,u]\leq B(x)+c$ which implies that the function $B(x)$ is a control barrier certificate. This concludes the proof.   
 \end{proof}
 \begin{remark}
 Assumption~\ref{ass:BC} is essential for applying the results of Lemma~\ref{sos} to \emph{any} LTL$_F$ specification. For a given specification, we can relax this assumption and allow some of the partition sets $X_i$ to be unbounded. For this, we require that the labels corresponding to unbounded partition sets should only appear either on self-loops or on accepting runs of length less than 3. For instance, Example~1 has an unbounded partition set $X_3$ and its corresponding label $p_3$ satisfies this requirement (see Figure 1), thus the results are still applicable.
 \end{remark}

 Based on Lemma~\ref{sos}, for any $\nu \in \mathcal P(\mathcal{A}_{\neg\varphi})$, a polynomial control barrier certificate $B_{\nu}(x)$ and controller $u_{\nu}(x)$ as in \eqref{eq_SCH_controller} can be computed using SOSTOOLS \cite{1184594} in conjunction with a semidefinite programming solver such as SeDuMi \cite{sturm1999using}. The computed barrier certificate will satisfy conditions in Theorem~\ref{barrier} while minimizing constants $\gamma_{\nu}$ and $c_{\nu}$.
Having values of $\gamma_{\nu}$ and $c_{\nu}$ for all $\nu \in \mathcal P(\mathcal{A}_{\neg\varphi})$, one can simply utilize results of Theorem~\ref{upper_bound} and Corollary~\ref{lower_bound} to compute a lower bound on the probability of satisfying the given specification to check the solution to Problem~\ref{SCS_prob}. 
{\begin{remark}\label{remark_minimize}
		To minimize the values of $\gamma_{\nu}$ and $c_{\nu}$ for each $\nu \in \mathcal P(\mathcal{A}_{\neg\varphi})$, one can simply utilize the bisection procedure by iteratively fixing $\gamma_{\nu}$ and minimizing over $c_\nu$ and then fixing the obtained $c_\nu$ and minimizing over $\gamma_{\nu}$. In this way, we give priority to minimizing $c_\nu$ to obtain a tight upper bound $(\gamma_{\nu}+c_\nu T_d) $ which is less sensitive to the finite time horizon $T_d$.
\end{remark}}
\begin{remark}
The procedure discussed above may result in a more conservative probability bounds due to the computation of common control barrier certificate in some cases. To obtain less conservative bounds one can simply substitute the constructed control policy in dynamics of the system and recompute barrier certificates minimizing constants $\gamma_{\nu}$ and $c_{\nu}$ for each $\nu \in \mathcal P(\mathcal{A}_{\neg\varphi})$ using Lemma~\ref{sos}. Then utilize these values to compute $\underline{\vartheta}$ in Problem~\ref{SCS_prob} using Theorem~\ref{upper_bound} and Corollary~\ref{lower_bound}. 
\end{remark}


\def\example{\par\noindent{\bf Example 1.} \ignorespaces}
\def\endexaple{}

\begin{example}(continued)
To compute control policy $u_\nu(x)$ and values of $\gamma_\nu$ and $c_\nu$ for each $\nu\in\mathcal{P}(\mathcal{A_{\neg\varphi}})$, we use SOS optimization according to Lemma~\ref{sos}  and minimize values of $\gamma$ and $c$ using bisection method. The optimization problem is solved using SOSTOOLS and SeDuMi. We choose barrier certificates $B$, SOS polynomials $\lambda_0, \lambda_1,\lambda$, and polynomial controller $\lambda_u$ of orders $4$, $2$, $2$, $2$ and $2$, respectively. The obtained controllers $u_\nu(x)$ and values of $\gamma_\nu$ and $c_\nu$ are listed in Table~\ref{table_1}. Now using Theorem~\ref{upper_bound}, one gets
\begin{align*}
\mathbb{P}^{x_0}_\rho\{L(\textbf{x}_N) \models\neg \varphi\}&\leq 4.883e\text{-}4 \times 0.002 + 4.883e\text{-}4  \times 9.766e\text{-}4  = 1.453e\text{-}6 , \text{ for all }x_0\in L^{-1}(p_0); \\
\mathbb{P}^{x_0}_\rho\{L(\textbf{x}_N) \models\neg \varphi\}&\leq 9.766e\text{-}4 ,  \text{ for all }x_0\in L^{-1}(p_2); \text{ and }\\
\mathbb{P}^{x_0}_\rho\{L(\textbf{x}_N) \models\neg \varphi\}&=1, \text{ for all }x_0\in L^{-1}(p_1)\cup L^{-1}(p_3).
\end{align*}
The control policy is given by $\rho(x,q_{\mathfrak m})=u_{\mu_{(q_{\mathfrak m}')}}(x)$, where $(q_{\mathfrak m},L(x),q_{\mathfrak m}')\in\delta_M$ is a transition in DFA $\mathcal{A}_{\mathfrak m}$ shown in Figure~\ref{SCH_swithing}. 
\begin{table}[]
	\caption{Controllers $u_\nu(x)$, constants $\gamma_\nu$, and $c_\nu$ for all $\nu\in\mathcal{P}(\mathcal{A_{\neg\varphi}})$, where $c_\nu=0$.}
	\label{table_1}
	\begin{tabular}{@{}lll@{}}
		\toprule
		$\mu_{(q,q',\Delta(q'))}$               & \multicolumn{1}{c}{\begin{tabular}[c]{@{}c@{}}$u_\nu(x)=a_0x_1^2+a_1x_1x_2+a_2x_1+a_3x^2+a_4x_2+a_5$\\ $[a_0,a_1,a_2,a_3,a_4,a_5]$\end{tabular}} & \multicolumn{1}{c}{$\gamma_\nu$} \\ \midrule
		$\{(q_0,q_1,q_2,3),(q_0,q_1,q_4,3)\}$ & [1.745e-3, 3.664e-6, $-$1.884e-4, 1.938e-3, 3.886e-4, 0.161]                                                                                         & 4.883e-4                         \\
		$\{q_1,q_2,q_3,3\}$                   & [1.321e-3, 3.252e-5, 2.544e-4, 1.828e-3, 4.212e-3, 0.228]                                                                                          & 0.002                            \\
		$\{q_1,q_4,q_3,3\}$                   & [1.754e-3, $-$6.636e-6, 1.636e-4, 1.934e-3, $-$2.170e-3, 0.163]                                                                                         & 9.766e-4                         \\
		$\{q_0,q_4,q_3,4\}$                   & [1.754e-3, 
		$-$6.636e-6, 1.636e-4, 1.934e-3, $-$2.170e-3, 0.163]                                                                                       & 9.766e-4                         \\ \bottomrule
	\end{tabular}
\end{table}
\qed
\end{example}

\subsubsection{Finite input sets}\label{CEGIS}
We use a counter-example guided inductive synthesis (CEGIS) framework to find control barrier certificates for the system $\mathfrak S$ with a finite input set $U$. The approach uses satisfiability (feasibility) solvers for finding barrier certificate of a given parametric form that handles quantified formulas by alternating between series of quantifier-free formulas using existing satisfiability modulo theories (SMT) solvers ({\it viz.},  Z3 \cite{Z3solver}, dReal \cite{dreal}, and OptiMathSAT \cite{sebastiani2015optimathsat}). In order to use CEGIS framework, we raise the following assumption.

\begin{assumption}\label{ass:BC1}
	System $\mathfrak S$ has a compact state set $X \subset \R^n$ and a finite input set $U=\{u_1,u_2,\ldots,u_l\}$, where $u_i\in\R^m$, $i\in\{1,2,\ldots,l\}$.
	Partition sets $X_i = L^{-1}(p_i)$, $i\in\{0,1,2,\ldots, M\}$, are bounded semi-algebraic sets. 
\end{assumption}
Under Assumption~\ref{ass:BC1}, we can formulate conditions of Theorem~\ref{barrier} as a satisfiability problem which can search for parametric control barrier certificate using CEGIS approach. The following Lemma gives a feasibility condition that is equivalent to conditions of Theorem~\ref{barrier}.
\begin{lemma}
	\label{sos1}
	Suppose Assumption~\ref{ass:BC1} holds and $X_0,X_1, X$ are bounded semi algebraic sets. Suppose there exists a function $B(x)$, constants $\gamma\in [0,1]$, and $c\ge 0$, such that following expression is true 
	\begin{align}
	\bigwedge_{x\in X} B(x)\geq 0 
	\bigwedge_{x\in X_0} B(x)\leq \gamma 
	\bigwedge_{x\in X_1} B(x)\geq1
	\bigwedge_{x\in X}\Big( \bigvee_{u\in U} (\mathbb{E}[B(f(x,u,w))\mid x,u]\leq B(x)+c) \Big). \label{eq:sos4}
	\end{align}
	Then, $B(x)$ satisfies conditions of Theorem~\ref{barrier} and any $u:X\rightarrow U$ with $u(x)\in \{u_i \in U\mid \mathbb{E}[B(f(x,u_i))\mid x,u_i]\leq B(x)+c\}$ is a corresponding control policy.
\end{lemma}
Now, we briefly explain the idea of CEGIS framework for computation of such a function $B(x)$.
\begin{itemize}
	\item[1.] Define a parameterized control barrier certificate of the form $B(p,x)=\sum_{i=1}^{\mathsf r}p_i b_i(x)$, where {basis functions $b_i(x)$ are monomials}, $p_i\in\R$ are unknown coefficients, and $i\in\{1,2,\ldots,\mathsf r\}$. 
	\item[2.] Select a finite set of samples $\overline{X}\subset X$, a constant $\gamma\in[0,1]$, and $c\ge 0$.
	\item[3.] Compute a candidate control barrier certificate $B(p,x)$ (\ie, coefficients $p_i$) such that the following expression is true.
	\begin{align*}
	\psi(p,x):=&\bigwedge_{x\in \overline{X}} B(p,x)\geq 0 
	\bigwedge_{x\in \overline{X}\cap X_0} B(p,x)\leq \gamma 
	\bigwedge_{x\in \overline{X}\cap X_1} B(p,x)\geq1\\
	&\bigwedge_{x\in \overline{X}}\Big( \bigvee_{u\in U} (\mathbb{E}[B(p,f(x,u,w))\mid x,u]\leq B(p,x)+c) \Big).
	\end{align*}
	The above expression results in linear arithmetic formula that involves boolean combinations of linear inequality constraints in $p_i$, which can be efficiently solved with the help of SMT solvers Z3 \cite{Z3solver} or OptiMathSAT \cite{sebastiani2015optimathsat}. 
	\item[4.] Search for a counter example $x_c\in X$ such that the candidate solution $B(p,x)$ obtained in the previous step satisfies $\neg \psi(p,x)$. Note that for a given $p$, satisfaction of $\neg \psi(p,x)$ is equivalent to the feasibility of a nonlinear constraint over $x$.
	If $\neg \psi(p,x)$ has no feasible solution, the obtained candidate solution is a true control barrier certificate for all $x\in X$ which terminates the algorithm.
	Otherwise, if $\neg \psi(p,x)$ is feasible for some $x=x_c\in X$, then we add that counter-example $x_c$ to the finite set, $\overline{X}:=\overline{X}\cup \{x_c\}$, and reiterate Steps 3--4. \\
	There are two possible ways to search for counter-examples:
	\begin{itemize}
		\item[(a)] \textit{Using SMT solvers}: To check satisfiability of $\neg \psi(p,x)$, one can use an SMT solver that can handle nonlinear constraints. For example, dReal \cite{dreal} is a general purpose nonlinear delta-satisfiability solver suitable for solving quantifier-free nonlinear constraints involving polynomials, trigonometric, and rational functions over compact sets $X$. We refer the interested readers to \cite{Ravanbakhsh2017ACO} for a more detailed discussion.
		\item[(b)] \textit{Using nonlinear optimization toolboxes}:
		To find counter-examples, one can alternatively solve a nonlinear optimization problem and check satisfaction of the following condition
		\begin{align*}
			&\text{If }\Big(\min_{x\in X}B(p,x)<0, \text{ OR } \min_{x\in X_0}-B(p,x)+\gamma<0, \text{ OR } \min_{x\in X_1}B(p,x)-1<0, \\
			&\quad\text{ OR } \min_{x\in X}\max_{u\in U}-\mathbb{E}[B(p,f(x,u,w))\mid x,u]+B(p,x)+c<0\Big)\\
			&\text{Then } x \text{ is a counter-example.} 
		\end{align*}
		
		To solve nonlinear optimization problems, one can use existing numerical optimization techniques such as sequential quadratic programming. Note that, the methods may run into local optima, however, one can utilize multi-start techniques \cite{Marta2003} to obtain global optima.
		For the final rigorous verification step, one can use tools like RSolver\footnote{http://rsolver.sourceforge.net} which extends a basic interval branch-and-bound method with interval constraint propagation. A detailed discussion on the verification algorithm used in RSolver can be found in \cite{ratschan2006efficient,ratschan2017simulation}.
	\end{itemize}
\end{itemize}
{This CEGIS algorithm is then iterated to minimize the values of $\gamma$ and $c$ in~\eqref{eq:sos4} as discussed in Remark~\ref{remark_minimize}. Note that, the CEGIS procedure either
	$(i)$ terminates after some finite iterations with a control barrier certificate satisfying \eqref{eq:sos4},
	$(ii)$ terminates with a counter example proving that no solution exists, or
	$(iii)$ runs forever. In order to guarantee termination of the algorithm, one can set an upper bound on the number of unsuccessful iterations.}
{\subsection{Computational Complexity}
	Characterizing the computational complexity of the proposed approaches is a very difficult task in general. However, in this subsection, we provide some analysis on the computational complexity.}

{From the construction of directed graph $\mathcal{G}=(\mathcal{V},\mathcal{E})$, explained in Section \ref{runs}, the number of triplets and hence the number of control barrier certificates needed to be computed are bounded by $|\mathcal{V}|^3 = |Q|^3$, where $|\mathcal{V}|$ is the number of vertices in $\mathcal{G}$. However, this is the worst-case bound. In practice, the number of control barrier certificates is much less. In particular, it is given by the number of all unique successive pairs of atomic propositions corresponding to the elements $\nu\in  P(\mathcal{A}_{\neg\varphi})$. 
	Further, it is known that $|Q|$ is at most $|\neg\varphi |2^{|\neg\varphi |}$, where $|\neg\varphi |$ is the length of formula $\neg\varphi$ in terms of number of operations \cite{baier2008principles}, but in practice, it is much smaller than this bound \cite{klein2006experiments}.}

{In the case of sum-of-squares optimization, the computational complexity of finding polynomials $B,\lambda_0,\lambda_1,\lambda_{u_i},$ and $\lambda_x$ in Lemma~\ref{sos} depends on both the degree of polynomials appearing in \eqref{eq:sos1}-\eqref{eq:sos2} and the number of state variables. It is shown that for fixed degrees, the required computations grow polynomially with respect to the dimension \cite{7364197}. Hence, we expect that this technique is more scalable in comparison with the discretization-based approaches, especially for large-dimensional systems. For the CEGIS approach, due to its iterative nature and lack of guarantee on termination, it is difficult to provide any analysis on the computational complexity.}

\section{Case Studies}\label{sec_case}
In this section, we consider two case studies to demonstrate the effectiveness of our results.  
\subsection{Temperature control of a room}
 We consider evolution of a room temperature given by stochastic difference equation 
\begin{align}\label{temp_model}
x(k+1)=x(k)+\tau_s(\alpha_{e}(T_e-x(k))+\alpha_{H}(T_h-x(k))u(k))+0.1 w(k),
\end{align} 
where $x(k)$ denotes the temperature of the room, $u(k)$ represents ratio of the heater valve being open, $w(k)$ is a standard normal random variable that models environmental uncertainties, $\tau_s=5$ minutes is the sampling time, $T_h=55^{\circ}C$ is the heater temperature, $T_e=15^{\circ}C$ is the ambient temperature, and  $\alpha_{e}=8\times10^{-3}$ and $\alpha_{H}=3.6\times10^{-3}$ are heat exchange coefficients. All the parameters are adopted from \cite{jagtap2017quest}. \\
The state set of the system is $X\subseteq \R$. We consider regions of interest $X_0=[21, 22]$,  $X_1=[0,20]$, $X_2=[23,45]$, and $X_3=X\setminus(X_0\cup X_1\cup X_2)$. The set of atomic propositions is given by $\Pi=\{p_0,p_1,p_2,p_3\}$ with labeling function $L(x_i)=p_i$ for all $x_i\in X_i$, $i\in\{0,1,2,3\}$. The objective is to compute a control policy with a potentially tight lower bound on the probability that the state evolution of length $N= 50$ satisfies the LTL$_F$ formula $\varphi=p_0\wedge\square\neg(p_1\vee p_2)$. The DFA $\mathcal{A}_{\neg\varphi}$ corresponding to $\neg\varphi$ is shown in Figure~\ref{DFA3}. One can readily see that, we have sets $\mathcal{P}^{p_0}=\{(q_0,q_1,q_2,49)\}$ and $\mathcal{P}^{p_1}=\mathcal{P}^{p_2}=\mathcal{P}^{p_3}=\emptyset$.
Next, we discuss the computational results for two cases of finite and continuous input sets.  
\begin{figure}[t] 
	\centering
	\includegraphics[scale=0.6, height = 2.8cm]{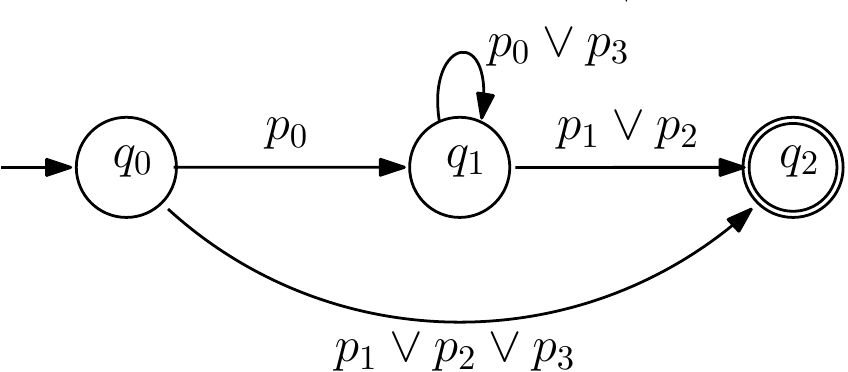} 
	\caption{DFA $\mathcal{A}_{\neg\varphi}$ that accept all traces of $\neg\varphi$, where $\varphi=p_0\wedge\square\neg(p_1\vee p_2)$.}
	\label{DFA3}
\end{figure}
\subsubsection{Finite input set.}
We consider that the control input $u(k)$ takes value in the set $U=\{0,0.5,1 \}$ (the heater valve is either closed, half open, or full open) and the temperature lies in the bounded set $X=[0,45]$.
We compute a control barrier certificate of order $4$ using the CEGIS approach discussed in Subsection~\ref{CEGIS} as the following:
$$B(x)=0.2167x^4-18.6242x^3+6.0032e2x^2-8.5998e3 x+4.6196e4.$$
The corresponding control policy is
\begin{equation}\label{contr_2}
u(x)= \min \{u_i \in U\mid \mathbb{E}[B(f(x,u_i))\mid x,u_i]\leq B(x)+c\}.
\end{equation}
%
One can readily see that the DFA of switching mechanism $\mathcal{A}_{\mathfrak m}$ contains only three states $Q_{\mathfrak m}=\{(q_0,\Delta(q_0)),$ $(q_0,q_1,\Delta(q_1)), q_2 \}$, thus we have control policy $\rho(x,q_{\mathfrak m})\equiv u(x)$. The lower bound $\mathbb{P}^{x_0}_\rho\{L(\textbf{x}_N) \models \varphi\} \geq 0.9766$ for all $x_0\in L^{-1}(p_0)$ is obtained using SMT solver Z3 and employing sequential quadratic programming for computing counterexamples as described in Subsection~\ref{CEGIS}. Values of $\gamma$ and $c$ are obtained as 0.008313 and 0.0003125, respectively. The implementation performed using Z3 SMT solver along with sequantial quadratic program in Python on an iMac (3.5 GHz Intel Core i7 processor) and it took around 4 minutes to find a control barrier certificate and the associated lower bound. Figure~\ref{cproperties} depicts the barrier certificate and the corresponding conditions in Theorem~\ref{barrier}: condition~\eqref{ine_1} is shown in a snippet in the top figure, condition~\eqref{ine_2} is shown in the top figure, and condition~\eqref{CBC_def1} for the control barrier certificate with control input $u(x)$ is shown in the bottom figure.
Figure~\ref{controller_u} presents the control policy $u:X\rightarrow U$ in \eqref{contr_2} and 
Figure~\ref{response1}  shows a few realizations of the temperature under this policy.
\begin{figure}[t] 
	\centering
	\includegraphics[scale=0.6]{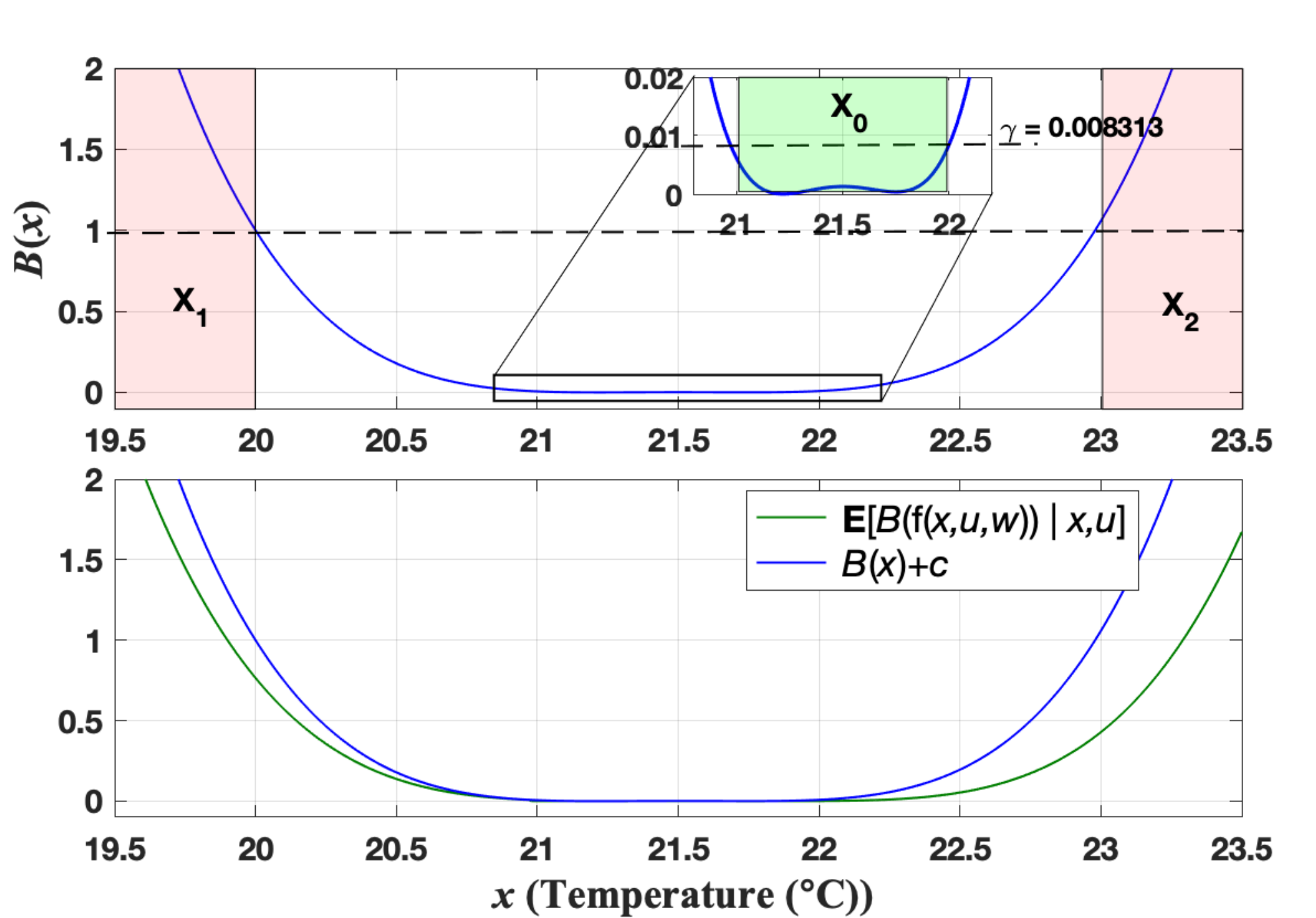} 
	\caption{Room temperature control: barrier certificate and the associated conditions from Theorem~\ref{barrier}. Condition \eqref{ine_1} is shown in the snippet in the top figure, condition \eqref{ine_2} is shown in the top figure, and condition~\eqref{CBC_def1} for the control barrier certificate under policy $u(x)$ is shown in the bottom figure.}
	\label{cproperties}
\end{figure}
\begin{figure}[t] 
	\centering
	\includegraphics[scale=0.56]{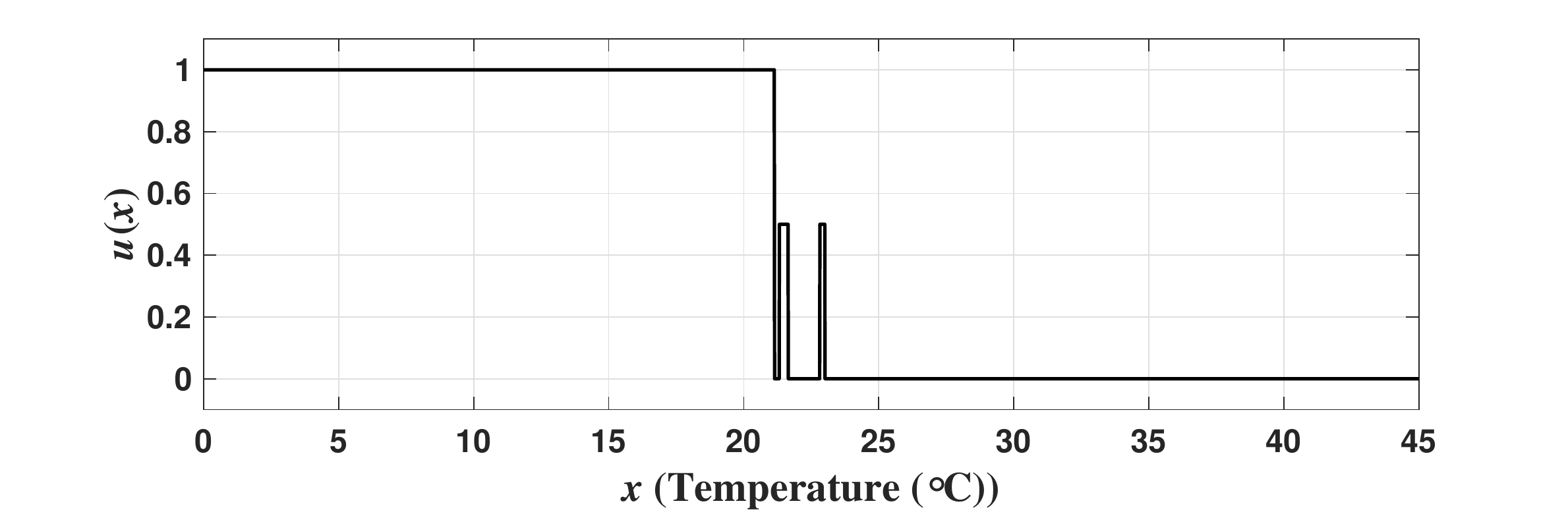} 
	\caption{Room temperature control: control policy $u: X\rightarrow \{0,0.5,1\}$ as given in \eqref{contr_2}.}
	\label{controller_u}
\end{figure}
\begin{figure}[t] 
	\centering
	\includegraphics[scale=0.56]{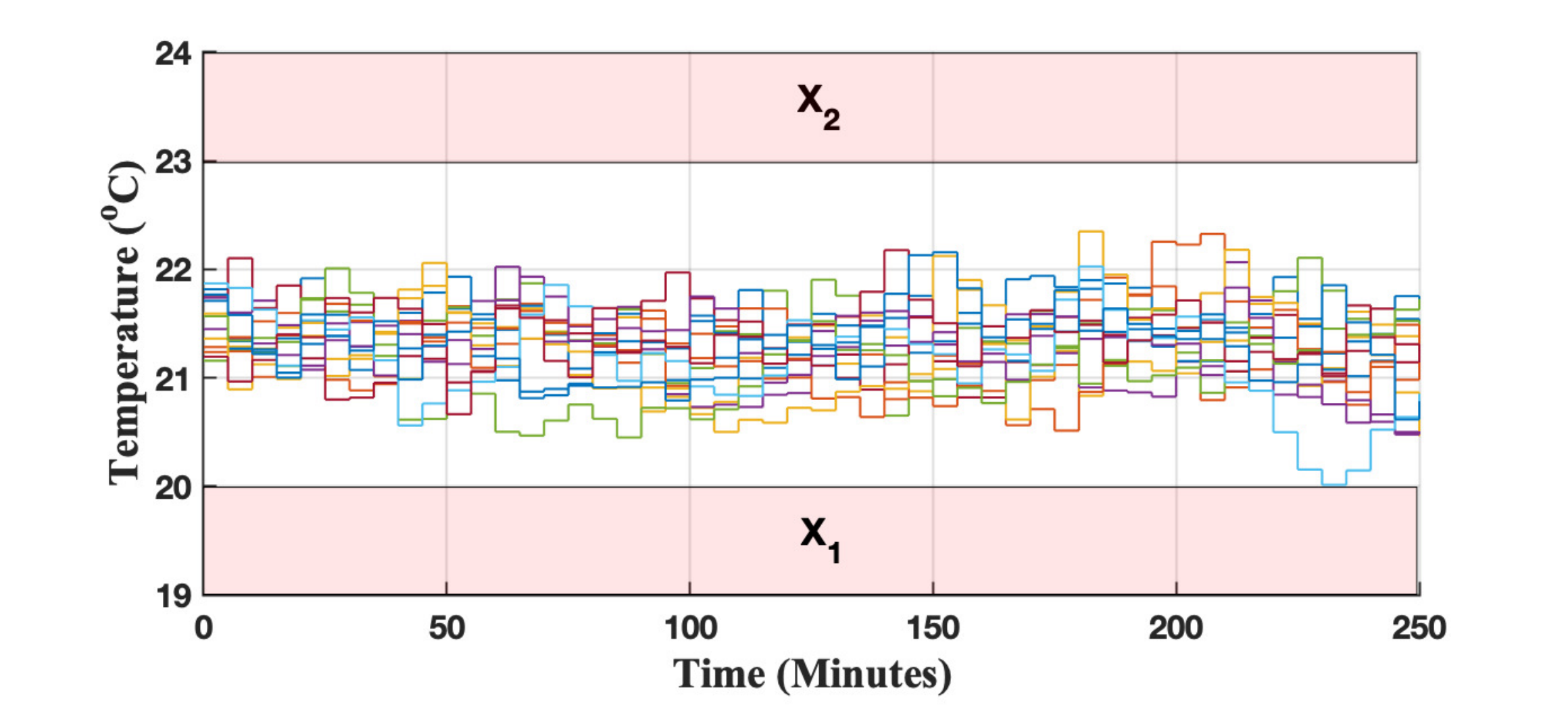} 
	\caption{Room temperature control: temperature evolution under control policy in \eqref{contr_2}.}
	\label{response1}
\end{figure}

\subsubsection{Continuous input set.}
Let us assume the system has the state space $X=\R$ and the continuous input set $U=[0,1]$ (the heater valve can be positioned continuously from fully closed to fully open). As described in Subsection~\ref{SOSS}, using Lemma~\ref{sos} we compute a control barrier certificate of order 4 as follows
$$B(x)=0.1911x_1^4-16.4779x_1^3+532.6393x_1^2-7651.3308x_1+41212.3666,$$
and the corresponding control policy of order 4 as
\begin{equation} \label{ccc}
u(x)=-1.018e\text{-}6x^4+7.563e\text{-}5x^3-0.001872x^2+0.02022x+0.3944.
\end{equation}
The values $\gamma=0.015625$, $c=0.00125$, and the lower bound $\mathbb{P}^{x_0}_\rho\{L(\textbf{x}_N) \models\varphi\} \geq 0.9281$ is obtained using SOSTOOLS and SeDuMi for all $x_0\in L^{-1}(p_0)$, as discussed in Subsection~\ref{SOSS}.
{The bound in this case is more conservative than the previous case with a finite input set. This is mainly due to the optimization algorithm that assumes fixed-degree polynomials $B(\cdot)$,
	$\lambda_0(\cdot)$, $\lambda_1(\cdot)$, $\lambda_x(\cdot)$, and $\lambda_u(\cdot)$.
	The computed lower bound can be improved by increasing the polynomial degrees but will result in a larger computational cost.}
The control policy and a few realizations of the temperature under this policy are shown in figures \ref{contr_1} and \ref{response}, respectively. 

\begin{figure}[t] 
	\centering
	\includegraphics[scale=0.56]{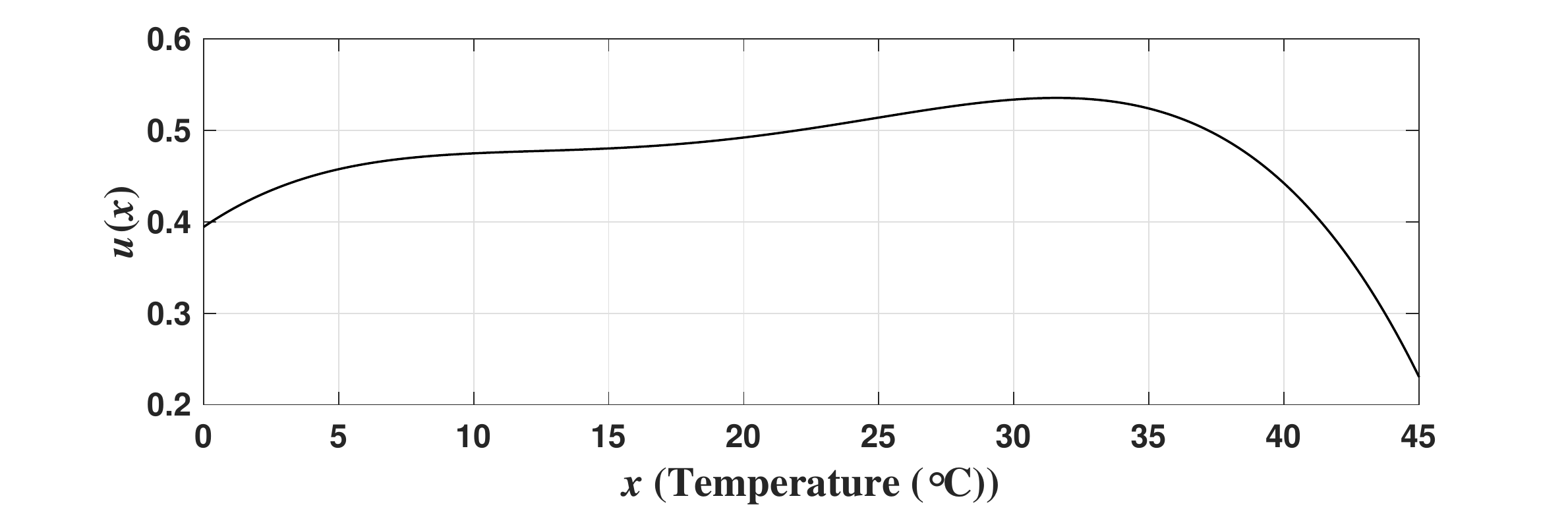} 
	\caption{Room temperature control: control policy $u: X\rightarrow [0,1]$ as given in \eqref{ccc}.}
	\label{contr_1}
\end{figure}

\begin{figure}[t] 
			\centering
			\includegraphics[scale=0.56]{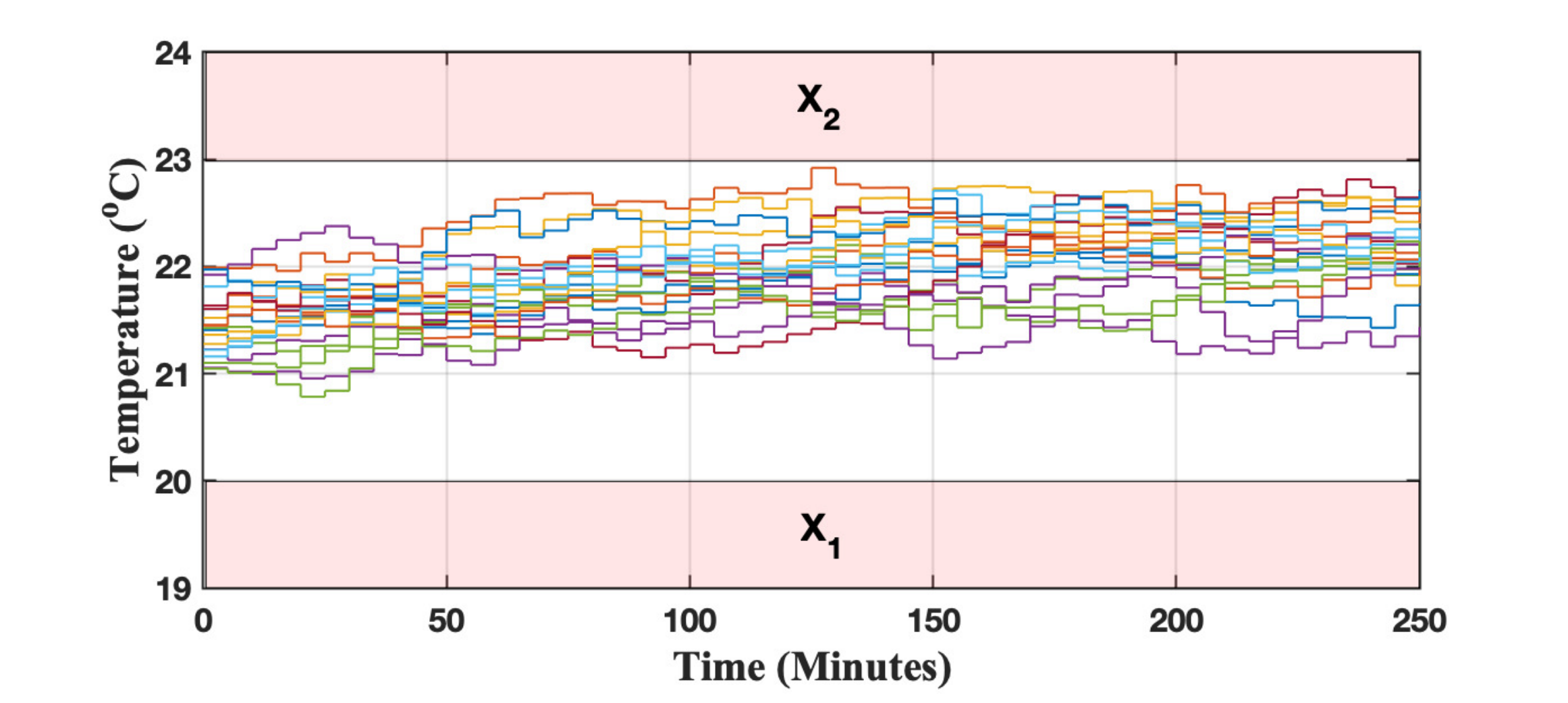} 
			\caption{Room temperature control: temperature evolution under control policy in \eqref{ccc}.}
			\label{response}
\end{figure}

{Discretization-based approaches provide a policy that is generally time-dependent. So it is not possible to directly compare our approach with them. However, using these techniques, we can validate the lower bound provided by our approach a posteriori.
	For this purpose, we combine our synthesized policy with the system to obtain an autonomous system and then use the toolbox FAUST$^2$ \cite{Faust} that computes an interval for the probability based on finite abstractions of the system. 
	The toolbox takes around $4$ minutes to verify the system using $314$ abstract states.
	The probability satisfies
	$$\mathbb{P}^{x_0}_\rho\{L(\textbf{x}_N) \models\varphi\}\in [1-5.458\times 10^{-4}, 1-3.612\times 10^{-4}],$$
	for all $x_0\in L^{-1}(p_0)$,
	which confirms the lower bound provided by our approach. For the purpose of comparison, we} { run the example using FAUST$^2$ to synthesize a time-dependent policy. In this case, the toolbox takes around 7 minutes and provides probability interval as 
	$$\mathbb{P}^{x_0}_\rho\{L(\textbf{x}_N) \models\varphi\}\in [1-1.5434\times 10^{-7},1-4.277\times 10^{-8}],$$ for all $x_0\in L^{-1}(p_0)$.}




\subsection{Lane keeping of a vehicle}
For the second case study, we consider a kinematic single-track model of a vehicle, specifically, BMW 320i, adopted from \cite{Althoff2017a} by discretizing the model and adding noises to capture the effect of uneven road. The corresponding nonlinear stochastic difference equation is
\begin{align*}
x_1(k+1) &= x_1(k)+\tau_s v \cos(x_4(k))+0.1w_1(k)\\
x_2(k+1) &= x_2(k)+\tau_s v \sin(x_4(k))+0.01w_2(k)\\
x_3(k+1) &= x_3(k)+\tau_s u(k)\\
x_4(k+1) &= x_4(k)+\frac{\tau_s v}{l_{wb}}\tan(x_3(k))+0.0005w_3(k),
\end{align*}  
where states $x_1, x_2,x_3,$ and $x_4$ represent $x$, $y$, the steering angle $\delta$, and the heading angle $\Psi$, respectively. The schematic showing states in the single-track model is shown in Figure~\ref{car_model}. The control input representing steering velocity is denoted by $u$. The terms $w_1$, $w_2$, and $w_3$ are noises in position and heading generated due to uneven road modeled using standard normal distribution. The parameters $\tau_s=0.01 s$, $l_{wb}=2.578 m$, and $v=10 m/s$ represent the sampling time, the wheelbase, and velocity, respectively.\\
We consider the state set $X=[0,50]\times[-6,6]\times[-0.05,0.05]\times[-0.1,0.1]$, finite input set $U=\{-0.5,0,0.5\}$, regions of interest $X_0=[0,5]\times[-0.1,0.1]\times[-0.005,0.005]\times[-0.05,0.05]$, $X_1=[0,50]\times[-6,-2]\times[-0.05,0.05]\times[-0.1,0.1]$, $X_1=[0,50]\times[2,6]\times[-0.05,0.05]\times[-0.1,0.1]$, and $X_3=X\setminus(X_0\cup X_1\cup X_2)$. The set of atomic propositions is given by $\Pi=\{p_0,p_1,p_2,p_3\}$ with labeling function $L(x_i)=p_i$ for all $x_i\in X_i$, $i\in\{0,1,2,3\}$. Our goal is to design a control policy to keep the vehicle in the middle lane for the time horizon of 4 seconds (\ie, $N=400$). The specification can be written as an LTL$_F$ formula $\varphi=p_0\wedge\square\neg(p_1\vee p_2)$.
Using CEGIS approach discussed in Subsection~\ref{CEGIS}, we compute a control barrier certificate as the following:
\begin{align*}
B(x)=&2.1794e\text{-}6 x_1^2+6.2500e\text{-}2 x_2^2-15.3131 x_3^2+1.0363 x_4^2+1.3088e\text{-}4 x_1\\&-4.4330e\text{-}5 x_2+0.3592 x_3-0.2488 x_4+5.9126e\text{-}2,
\end{align*}
and the corresponding control policy as
\begin{equation}\label{contr_3}
u(x)\in \{u_i \in U\mid \mathbb{E}[B(f(x,u_i))\mid x,u_i]\leq B(x)+c\},
\end{equation}
which guarantees $\mathbb{P}^{x_0}_\rho\{L(\textbf{x}_N) \models \varphi\}\ge 0.8688$ with values $\gamma=0.03125$ and $c=0.00025$. Figure~\ref{response_car} shows a few realizations of the system under the control policy \eqref{contr_3}. The implementation performed using the Z3 SMT solver along with the sequential quadratic program in Python on an iMac (3.5 GHz Intel Core i7 processor) and it took around 30 hours to find a control barrier certificate and the associated lower bound. 
Note that, since the procedure described in Subsection~\ref{CEGIS} is highly parallelizable, the execution time can be reduced significantly. 
Note that due to the large dimension of the state set, FAUST$^2$ is not able to give a lower bound on the probability of satisfaction. {However, for the sake comparison, we employ the Monte-Carlo approach to obtain the empirical probability interval as $\mathbb{P}^{x_0}_\rho\{L(\textbf{x}_N) \models \varphi\}\in [0.9202, 0.9630]$ with the confidence $1-10^{-10}$ using $10^5$ realizations with the controller in \eqref{contr_3}, which confirms the lower bound obtained using our approach.}
\begin{figure}[t] 
	\centering
	\includegraphics[scale=0.7]{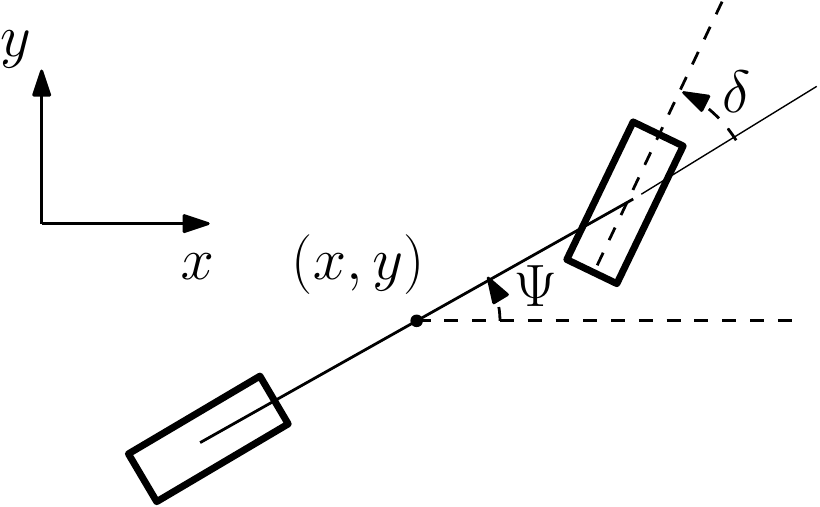} 
	\caption{Single-track model}
	\label{car_model}
\end{figure}     
\begin{figure}[t] 
	\centering
	\includegraphics[scale=0.56]{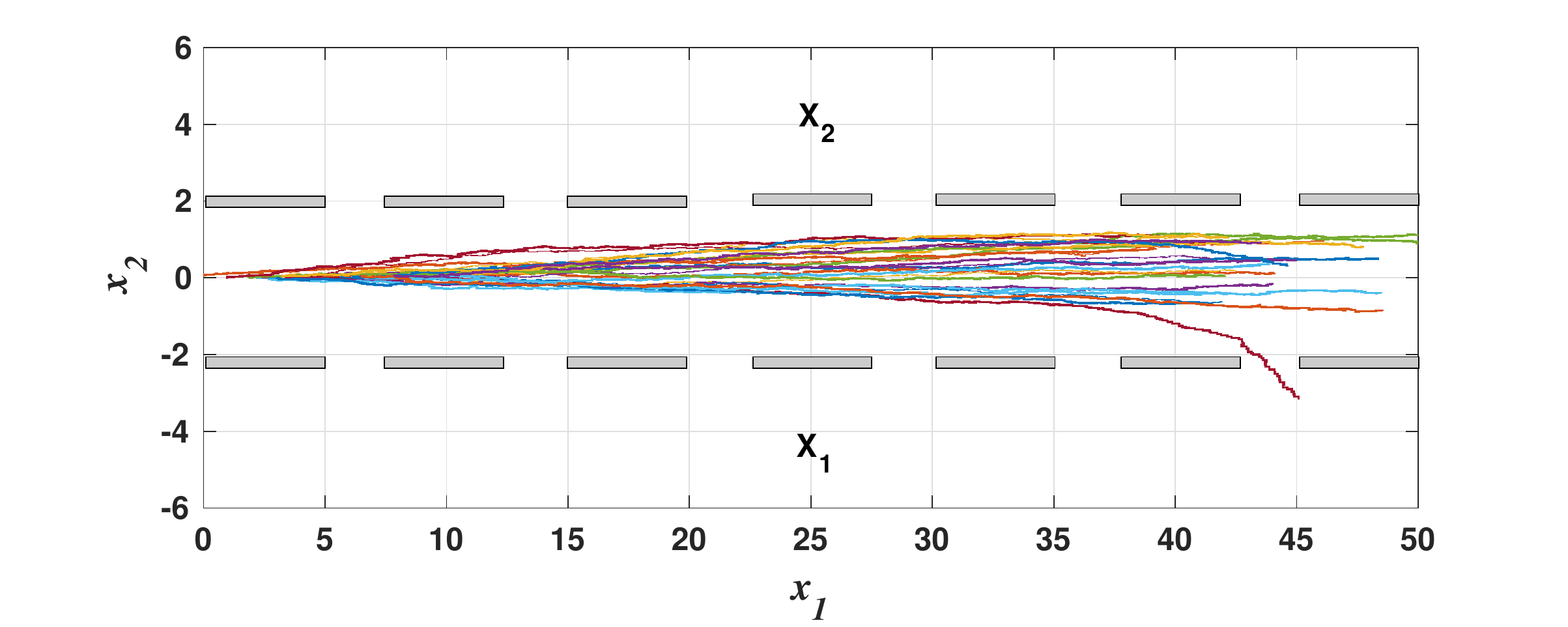} 
	\caption{Several closed-loop realization using controller in \eqref{ccc}.}
	\label{response_car}
\end{figure}
\section{Conclusion} \label{sec_conclusion}
In this paper, we proposed a discretization-free approach for the formal synthesis of discrete-time stochastic control systems. The approach computes a control policy together with a lower bound on the probability of satisfying a specification encoded as LTL over finite traces. It utilizes computation of control barrier certificates and uses sum-of-squares optimization or counter-example guided inductive synthesis to obtain such policies. {Currently, our approach is restricted to LTL$_F$ properties and is computationally applicable only to systems with dynamics that can be transformed into polynomial (in)equalities. The outcome of the approach is also restricted to stationary policies.}

Our approach can easily be extended to synthesize policies for continuous-time stochastic control systems enforcing LTL$_F$ specifications by excluding the next operator. The results may become more conservative in this case since an efficient computation of the temporal horizon $T(\cdot)$ as in~\eqref{eq:reachability} is not possible and one needs to consider the worst-case $T=N$. Although the proposed approach seems scalable in comparison with the discretization-based ones, we are actively working on improving scalability further by providing a compositional construction of control barrier certificates for large-scale systems (see \cite{jagtap2020compositional} for our recent work providing compositional construction of control barrier certificates for non-stochastic interconnected systems).
From the implementation point of view, we plan to provide an efficient toolbox leveraging parallel computations for solving these synthesis problems.

\bibliographystyle{alpha}
\bibliography{IEEEtran1}

\end{document}